\documentclass[a4paper,11pt]{article}
\RequirePackage{natbib}

\usepackage{amsmath}
\usepackage{amssymb}
\usepackage{latexsym}
\usepackage{ntheorem}
\usepackage{ifthen}
\usepackage{xcolor}
\usepackage{newcent}
\usepackage {mathtools}
\mathtoolsset{showonlyrefs=true}

\usepackage[margin=3cm]{geometry}

\theoremstyle{plain}

\theoremheaderfont{\bf}\theorembodyfont{\sl}

\newtheorem{theorem}{Theorem}[section]
\newtheorem{proposition}[theorem]{Proposition}
\newtheorem{lemma}[theorem]{Lemma}

\theoremheaderfont{\it}\theorembodyfont{\rm}

\newtheorem{remark}[theorem]{Remark}

\theoremheaderfont{\bf}\theorembodyfont{\rm}

\newtheorem{definition}[theorem]{Definition}
\newtheorem{example}[theorem]{Example}

\newtheorem{assumption}[theorem]{Assumption}

\theoremstyle{nonumberplain}

\theoremheaderfont{\bf}\theorembodyfont{\sl}

\newenvironment{proof}[1][]
{\ifthenelse{\equal{#1}{}}{\smallskip\noindent\textsl{Proof. }}{\smallskip
\noindent\textsl{Proof #1. }}}{\hfill$\Box$}

\def\Bc{{\cal B}}
\def\Hc{{\cal H}}

\def\Gc{{\cal G}}
\def\Dc{{\cal D}}

\def \l{\lambda}

\def \o{\omega}
\def \O{\Omega}

\begin{document}
\title{Multiple-priors Optimal Investment in Discrete Time for Unbounded Utility Function}
\author {
{Romain} {Blanchard}, E.mail~:  romain.blanchard@etudiant.univ-reims.fr\\
LMR, Universit\'e Reims Champagne-Ardenne.\\
\and
{Laurence} {Carassus}, E.mail~: laurence.carassus@devinci.fr \\
Research Center, L\'eonard de Vinci P\^ole universitaire \\
and LMR, Universit\'e Reims Champagne-Ardenne.\\
}

\date{}

\maketitle

\begin{abstract}
This paper investigates the problem of maximizing expected terminal utility in a discrete-time financial market model with a finite horizon under non-dominated model uncertainty. We use a dynamic programming framework together with measurable selection arguments to prove  that under mild integrability conditions, an optimal portfolio exists for an unbounded utility function defined on the half-real line.
\end{abstract}
\textbf{Key words}: {Knightian uncertainty; multiple-priors; non-dominated model; optimal investment} \\
\textbf{AMS 2000 subject classification}: {Primary 93E20, 91B70, 91B16 ; secondary
 91G10, 28B20, 49L20}

\section{Introduction}
We consider investors trading in a multi-period and discrete-time financial market. We study the  problem  of  terminal wealth expected utility  maximisation under Knightian uncertainty. It was  first introduced by F. Knight \citep{Kni} and refers to the ``unknown unknown", or uncertainty, as opposed to the ``known unknown", or risk.  This concept is very appropriate  in the context of financial mathematics as it describes accurately market behaviors which are becoming more and more surprising. The belief of investors are modeled with a set of probability measures rather than a single one. This can be related to model mispecification issues or model risk and has triggered a renewed and strong interest by practitioners and academics alike.  \\
 The axiomatic theory of the classical  expected utility was initiated by \citep{vNM}. They provided conditions on investor preferences under which the expected utility of a contingent claim $X$ can be expressed as  $E_{P} U(X)$ where $P$ is  a given  probability measure  and $U$ is a so-called utility function. The problem of  maximising the von Neumann and Morgenstern expected utility   has been extensively  studied, we refer to  \citep{RS05} and \citep{RS06} for the discrete-time case and  to \citep{KS99} and \citep{S01} for  the continuous-time one.  In the presence of Knightian uncertainty, \citep{Gilb} provided a pioneering contribution by extending the axiomatic of von Neumann and Morgenstern.  In this case, under suitable conditions on the investor preferences,  the utility functional is of the form $\inf_{P \in \mathcal{Q}^T} E_{P} U(X)$ where $\mathcal{Q}^T$ is the set of all  possible probability measures representing  the agent  beliefs.  Most of the literature on the so-called multiple-priors or robust expected utility maximisation assumes that $\mathcal{Q}^T$ is dominated by a reference measure. We refer to \citep{FSW09} for  an extensive survey.\\
However assuming the existence of a dominating reference measure  does not always provide the required degree of generality from an economic and practical perspective.  Indeed,  uncertain volatility models (see \citep{AP95},  \citep{DM06}, \citep{Ly95}) are  concrete examples where this hypothesis fails.  On the other hand, assuming a  non-dominated  set of probability measures significantly raises the mathematical difficulty of the problem as some of the usual tools of probability theory do not apply.
In the multiple-priors non-dominated  case,  \citep{DenKer} obtained the existence of an optimal strategy, a worst case measure as well as  some ``minmax" results under some compacity assumption on the set of probability measures and with a bounded (from above and below) utility function. This result is obtained in the continuous-time case. In  the discrete-time case, \citep{Nutz} (where further references  to multiple-priors non-dominated  problematic can be found)  obtained the first  existence result without any compacity assumption on the set of probability measures but  for a bounded (from above) utility function. We also mention two articles subsequent to our contribution. The first one (see \citep{Bart16}) provides a dual representation in the case of an exponential utility function with  a random endowment and the second one (see \citep{NS16}) study a market with frictions in the spirit of \citep{PennanenPerkio} for  a bounded from above utility function. \\
To the best of our knowledge, this paper provides the first general result for unbounded utility functions assuming a  non-dominated  set of probability measures (and without compacity assumption). This includes for example, the useful case of  Constant Relative Risk Aversion utility functions ($i.e$ logarithm or power functions). In Theorem \ref{main}, we give sufficient conditions for the existence of an optimizer to our ``maxmin" problem (see Definition \ref{DefU}). We  work  under the framework of \citep{BN} and \citep{Nutz}. The market is governed by a non-dominated set of probability measures $\mathcal{Q}^T$ that determines which events are relevant or not. Assumption \ref{Qanalytic}, which is related to measurability issues,  is the only assumption made on $\mathcal{Q}^T$ and is the cornerstone of the proof.  We introduce two integrability assumptions. The first one  (Assumption \ref{Uminus})  is  related to measurability and continuity issues. The second one (Assumption \ref{Uplus}) replaces the boundedness assumption of \citep{Nutz} and allows us to use auxiliary  functions which play the role of properly integrable bounds for the value functions at each step.
The no-arbitrage condition is essential as well, we use the one introduced in \citep{BN} and propose  a ``quantitative" characterisation  in the spirit of \citep{JS98} and \citep{RS05}.
Finally,  we introduce an  alternative ``strong" no-arbitrage condition (the $sNA$, see Definition \ref{sNA})  and  prove in Theorem \ref{main2} that under the $sNA$ condition, Theorem \ref{main} applies to a large range of settings.  \\
As in \citep{BN} and \citep{Nutz}  our proof relies heavily on measure theory tools, namely on analytic sets.
Those sets display the nice property of being stable by projection or countable unions and intersections. However they fail to be stable by complementation, hence the sigma-algebra generated by analytic sets contains sets that are not analytic  which leads to significant  measurability issues. Such difficulties arise for instance in Lemma \ref{Dhmes}, where we are still able to prove some tricky measurability properties, as well as  in Proposition \ref{dyn3} which is pivotal in solving the dynamic programming.
Note as well,  that we have identified (and corrected) a small issue    in \citep[Lemma 4.12]{BN} which is also used in \citep{Nutz} to prove some important measurability properties. Indeed it is not enough in order to have joint-measurability of a function $\theta(\o,x)$ to assume that $\theta(\cdot,x)$ is measurable and $\theta(\o,\cdot)$ is lower-semicontinuous, one has to assume for example that $\theta(\o,\cdot)$ is convex (see Lemma \ref{412} as well as the counterexample  \ref{uscusc}).
\\
To solve our optimisation problem we follow a similar approach as \citep{Nutz}. We first consider  a one-period case with strategy in $\mathbb{R}^{d}$. To ``glue" together the solutions found in the one-period case we  use dynamic programming  as in \citep{RS05}, \citep{RS06}, \citep{CR14}, \citep{CRR}, \citep{Nutz} and \citep{BCR} together with  measurable selection arguments (Auman and Jankov-von Neumann Theorems). \\
In the remainder of the introduction, we recall   some important properties of analytic sets, present our framework and state our main result. In section \ref{secna}  we prove our quantitative version of  the multiple-priors no-arbitrage condition.  In section \ref{se3}  we solve the expected utility maximisation problem, first in the one period case.
Finally,  section \ref{seannexe}  collects some technical results and proofs as well as some counter-examples to  \citep[Lemma 4.12]{BN}.

\subsection{\textbf{Polar sets and  universal sigma-algebra}}
For any Polish space $X$ ($i.e$  complete and separable metric space), we denote by $\mathcal{B} (X)$ its Borel sigma-algebra  and by $\mathfrak{P}(X)$ the set of all probability measures on $(X,\mathcal{B} (X))$. We recall  that $\mathfrak{P}(X)$ endowed with the weak topology is a Polish space (see  \citep[Propositions 7.20 p127, 7.23 p131]{BS}). If $P$ in $\mathfrak{P}(X)$, $\mathcal{B}_{P}(X)$ will be the completion of $\mathcal{B}(X)$ with respect to $P$ and the universal sigma-algebra is defined by $\mathcal{B}_{c}(X):= \bigcap_{P \in \mathfrak{P}(X)} \mathcal{B}_{P}(X)$. It is clear that $\mathcal{B}(X) \subset \mathcal{B}_{c}(X)$. In the rest of the paper we will use the same notation for
$P$ in $\mathfrak{P}(X)$ and for its (unique) extension on $\mathcal{B}_{c}(X)$. A function $f: X \to Y$ (where $Y$ is an other Polish space) is universally-measurable or $\mathcal{B}_{c}(X)$-measurable (resp. Borel-measurable or $\mathcal{B}(X)$-measurable) if for all $B \in \mathcal{B}(Y)$, $f^{-1}(B) \in  \mathcal{B}_{c}(X)$ (resp. $f^{-1}(B) \in  \mathcal{B}(X)$).  Similarly we will speak of universally-adapted or universally-predictable (resp. Borel-adapted or Borel-predictable) processes.

For a given $\mathcal{P} \subset  \mathfrak{P}(X)$, a set $N \subset X$ is called a $\mathcal{P}$-polar  if for all $P \in \mathcal{P}$, there exists some $A_{P} \in \mathcal{B}_c(X)$ such that $P(A_{P})=0$ and $N \subset A_{P}$. We say that a property holds true $\mathcal{P}$-quasi-surely (q.s.), if it is true outside a $\mathcal{P}$-polar set. Finally we say that a set is of $\mathcal{P}$-full measure  if its complement is a $\mathcal{P}$-polar set.

\subsection{\textbf{Analytic sets}}
 An analytic set of $X$ is the continuous image of a Polish space, see  \citep[Theorem 12.24 p447]{Hitch}. We denote by $\mathcal{A}(X)$ the set of analytic sets of $X$ and recall some key properties that  will often be  used in the rest of the paper without further references (see also \citep[Chapter 7]{BS} for more details on analytic sets). The projection of an analytic set is an analytic set  (see  \citep[Proposition 7.39 p165]{BS}) and the countable union, intersection or cartesian product of analytic sets is an analytic set (see \citep[Corollary 7.35.2 p160, Proposition 7.38 p165]{BS}). However the complement of an analytic set does not need to be an  analytic set. We denote by $\mathcal{C}A(X):=\left\{A \in X,\; X \backslash{A} \in \mathcal{A}(X)\right\}$ the set of all coanalytic sets of $X$. We have that (see  \citep[Proposition 7.36 p161, Corollary 7.42.1 p169]{BS})
 \begin{align}
 \label{analyticset}
 \mathcal{B}(X) \subset \mathcal{A}(X) \cap \mathcal{CA}(X) \; \mbox{and} \; \mathcal{A}(X) \cup \mathcal{CA}(X) \subset \mathcal{B}_{c}(X).
  \end{align}
Now, for $D \in \mathcal{A}(X)$, a function $f: D  \to \mathbb{R} \cup \{\pm \infty\}$ is lower-semianalytic or lsa (resp. upper-semianalytic or usa) on $X$ if  $\{x \in X\; f(x)<c\} \in \mathcal{A}(X)$ (resp.  $\{x \in X\; f(x)>c\} \in \mathcal{A}(X)$) for all $c \in \mathbb{R}$. We denote by $\mathcal{L}SA(X)$ (resp. $\mathcal{U}SA(X)$) the set of all lsa (resp. usa) functions on $X$. A function $f: X \to Y$ (where $Y$ is another Polish space) is analytically-measurable if for all $B \in \mathcal{B}(Y)$, $f^{-1}(B)$ belongs to the sigma-algebra generated by $\mathcal{A}(X)$.  From \eqref{analyticset} it is clear that  if $f$ is lsa or usa or analytically-measurable then $f$ is $\mathcal{B}_{c}(X)$-measurable, again this will be used through the paper without further references.

 \subsection{\textbf{Measurable spaces, stochastic kernels and definition of $\mathcal{Q}^T$ }}
We fix  a time horizon $T\in \mathbb{N}$ and  introduce  a sequence $\left(\Omega_t\right)_{1 \leq t \leq T}$  of Polish spaces. We denote by $\Omega^{t}:=\O_{1} \times \dots \times \O_{t},$  with the convention that $\Omega^{0}$ is reduced to a singleton. An element of $\Omega^{t}$ will be denoted by $\o^{t}=(\omega_{1},\dots, \omega_{t})=(\o^{t-1},\o_{t})$ for $(\o_{1},\dots,\o_{t}) \in \Omega_{1}\times\dots\times\Omega_{t}$ and $(\o^{t-1},\o_{t}) \in \Omega^{t-1}\times\Omega_{t}$ (to avoid heavy notation we drop the dependency in $\o_{0}$).
It is well know that  $\mathcal{B}(\Omega^{t})=\mathcal{B}(\Omega^{t-1}) \otimes \mathcal{B}(\Omega_{t})$, see  \citep[Theorem 4.44 p149]{Hitch}. However  we have only that $\mathcal{B}_{c}(\Omega^{t-1}) \otimes \mathcal{B}_{c}(\Omega_{t}) \subset \mathcal{B}_{c}(\Omega^{t})$, which makes the use of the Projection Theorem problematic and enlighten why analytic sets are introduced.
For all $0 \leq t \leq T-1$, we denote by $\mathcal{S}K_{t+1}$ the set of universally-measurable stochastic kernel on $\O_{t+1}$ given $\O^{t}$ (see \citep[Definition 7.12 p134, Lemma 7.28 p174]{BS} ).
Fix some $1\leq t\leq T$,  $P_{t-1} \in \mathfrak{P}(\O^{t-1})$ and $p_{t} \in \mathcal{S}K_{t}$. Using Fubini's Theorem, see  \citep[Proposition 7.45 p175]{BS}, we set for all $A \in \mathcal{B}_c(\Omega^{t})$
$$
P_{t-1}\otimes p_{t}(A):=\int_{\Omega^{t-1}} \int_{\Omega_{t}} 1_{A}(\omega^{t-1},\omega_{t}) p_{t}(d\omega_{t},\omega^{t-1}) P_{t-1}(d\omega^{t-1}).
$$
For all $0\leq t\leq T-1$, we consider the random sets $\mathcal{Q}_{t+1}: \Omega^t \twoheadrightarrow \mathfrak{P}(\O_{t+1})$:  $\mathcal{Q}_{t+1}(\o^{t})$  can be seen as the set of possible models for the $t+1$-th period given the state $\o^{t}$ until time $t$.
\begin{assumption}
\label{Qanalytic}
For all $0\leq t\leq T-1$, $\mathcal{Q}_{t+1}$ is a non-empty and convex valued random set such that
$
\mbox{Graph}(\mathcal{Q}_{t+1})=\left\{(\omega^{t},P),\; P \in \mathcal{Q}_{t+1}(\omega^{t})\right\} \in  \mathcal{A}\left(\Omega^{t} \times \mathfrak{P}(\Omega_{t+1})\right).
$
\end{assumption}

From the Jankov-von Neumann Theorem, see  \citep[Proposition 7.49 p182]{BS},  there exists some analytically-measurable $q_{t+1}: \Omega^{t} \to \mathfrak{P}(\O_{t+1})$ such that for all $\o^{t} \in \O^{t}$, $q_{t+1}(\cdot,\o^{t}) \in \mathcal{Q}_{t+1}(\o^{t})$ (recall that for all $\o^{t} \in \O^{t}$, $\mathcal{Q}_{t+1}(\o^{t}) \neq \emptyset$). In other words  $q_{t+1} \in \mathcal{S}K_{t+1}$ is a universally-measurable selector of  $\mathcal{Q}_{t+1}$.
For all $1\leq t \leq T$ we define $\mathcal{Q}^{t} \subset \mathfrak{P}\left(\Omega^{t}\right)$ by
\begin{footnotesize}
\begin{align}
\label{Qstar}
 \mathcal{Q}^{t}:=\{& Q_{1} \otimes q_{2} \otimes \dots \otimes q_{t},\;  Q_{1} \in \mathcal{Q}_{1},\;  q_{s+1} \in \mathcal{S}K_{s+1},
 q_{s+1}(\cdot,\o^{s}) \in \mathcal{Q}_{s+1}(\o^{s})\; Q_{s} \mbox{-a.s.} \; \, \forall \;  1\leq s\leq t-1\; \},
\end{align}
\end{footnotesize}
where if $Q_{t}=Q_{1}\otimes q_{2} \otimes \dots \otimes q_{t} \in \mathcal{Q}^{t}$ we write for any $2 \leq s \leq t $ $Q_{s}:=Q_{1}\otimes q_{2} \otimes \dots \otimes q_{s}$ and  $Q_{s}\in \mathcal{Q}^{s}$.
For any fixed $P \in \mathcal{Q}^T$,   $E_{P}$ denotes the expectation under $P$.
\subsection{\textbf{The traded assets and strategies}}
Let $S:=\left\{S_{t},\ 0\leq t\leq T\right\}$ be a universally-adapted $d$-dimensional  process
where for $0\leq t\leq T$, $S_{t}=\left(S^i_t\right)_{1 \leq i \leq d}$ represents the  price of $d$ risky securities in the
financial market in consideration. We make the following assumptions which were already stated in \citep{Nutz}.
\begin{assumption}
\label{Sass}
The process $S$ is  Borel-adapted.
 \end{assumption}
 \begin{remark}
If Assumption \ref{Sass} is not postulated,  we cannot obtain some crucial measurability properties (see  \citep[Remark  4.4]{BN}, Lemma \ref{Dmeasurability} below as well as \eqref{reims} and \eqref{reims2} and \citep[Lemma 7.30 (3) p178]{BS}). Note that this assumption is not needed in the one period case.
 \end{remark}
\begin{assumption}
\label{Sass2}
There exists some $0 \leq s <\infty$ such that $-s \leq S^{i}_{t}(\o^{t}) <+\infty $ for all $1 \leq i \leq d$, $\o^{t} \in \O^{t}$ and  $0 \leq t \leq T$.
 \end{assumption}
Note that we can easily incorporate the case where $-s \leq S^{i}_{t} <+\infty $ only on a
Borel  $\mathcal{Q}^T$-full measure set.
There exists also a riskless
asset for which we assume a price constant equal to $1$, for  sake of simplicity. Without this assumption, all the developments
below could be carried out using discounted prices. The notation $\Delta
S_t:=S_t-S_{t-1}$ will often be used.
If $x,y\in\mathbb{R}^d$ then
the concatenation $xy$ stands for their scalar product. The symbol $|\cdot|$ denotes the Euclidean norm
on $\mathbb{R}^d$ (or on $\mathbb{R})$. Trading
strategies are represented by $d$-dimensional universally-predictable processes $\phi:=\{ \phi_{t}, 1 \leq t \leq T\}$ where for all $1 \leq t \leq T$, $\phi_{t}=\left(\phi^{i}_{t}\right)_{1 \leq i \leq d}$ represents the
investor's holdings in  each of the $d$ assets at time $t$. The family of all such trading
strategies is denoted by $\Phi$.
We assume that trading is self-financing. As the riskless asset's price is constant equal to $1$, the value at time $t$ of a portfolio $\phi$ starting from
initial capital $x\in\mathbb{R}$ is given by $V^{x,\phi}_t=x+\sum_{s=1}^t  \phi_s \Delta S_s.$

From now on the positive (resp. negative) part of some number or
random variable $Y$ is denoted by $Y^+$ (resp. $Y^-$). We will also write $f^{\pm}(Y)$ for $\left(f(Y)\right)^{\pm}$ for any random
variable $Y$ and (possibly random) function $f$

\subsection{\textbf{No arbitrage condition, risk preferences and main result}}
\label{ref1}
\begin{definition}
     \label{NAQ}
     The $NA(\mathcal{Q}^{T})$ condition holds true if for $\phi  \in \Phi$,
$V_{T}^{0,\phi} \geq 0 \; \mathcal{Q}^{T}\mbox{-q.s.} $ implies that $V_{T}^{0,\phi}  = 0 \;\mathcal{Q}^{T}\mbox{-q.s.}$ (see also \citep[Definition 1.1]{BN}).
\end{definition}

\begin{definition}
\label{utilite}
A \emph{random utility} $U$ is a function  defined on $\Omega^T \times (0,\infty)$ taking values in $\mathbb{R}\cup \{-\infty\}$ such that for every $x\in \mathbb{R}$,  $U \left(\cdot,x\right)$  is $\mathcal{B}(\Omega^{T})$-measurable and for every $\o^{T} \in {\Omega}^{T}$, $U(\o^{T}, \cdot)$ is proper \footnote{There exists $x \in(0,+\infty)$ such that  $U(\o^{T}, x)>-\infty$ and $U(\o^{T}, x)<+\infty$ for all $x \in(0,+\infty)$.}, non-decreasing and  concave  on $(0,+\infty)$. We extend $U$ by (right) continuity in $0$ and set  $U(\cdot,x)=-\infty$ if $x<0$.
\end{definition}
\begin{remark}
\label{uscoupas}
Fix some $\o^{T} \in \O^{T}$ and let $ \mbox{Dom}  \, U(\o^{T},\cdot):=\{ x \in \mathbb{R}, U(\o^{T},x)>-\infty\}$ be the domain of $U(\o^{T},\cdot)$. Then $ U(\o^{T},\cdot)$ is continuous on $\mbox{Ri} (\mbox{Dom} \, U(\o^{T},\cdot))$, the relative interior of the domain of $U(\o^{T},\cdot)$ (see  \citep[Theorem 10.1 p82]{cvx}). Note that if $ U(\o^{T},\cdot)$ is improper then $ U(\o^{T},\cdot)=+\infty$ on $\mbox{Ri} (\mbox{Dom} \, U(\o^{T},\cdot))$  and if $ U(\o^{T},\cdot)$ is assumed to be upper semicontinuous (usc from now) then it is infinite on all $\mathbb{R}$ (see  \citep[Theorem 7.2 and Corollary 7.2.1, p53]{cvx}) which is a rather uninteresting case. Nevertheless our results hold true for an improper usc function. Here $U(\o^{T},\cdot)$ will not be assumed to be usc since Assumption \ref{Uminus} is postulated. Indeed it implies that $\mbox{Dom} \, U(\o^{T},\cdot)=(0,\infty)$ if  $\o^{T} \in \O_{Dom}^{T}$ which is a Borel  $\mathcal{Q}^{T}$-full measure set (see Lemma \ref{Ucontinuity}). Then $U$ can be modified so that it remains Borel-measurable, that $\mbox{Dom} \, U(\o^{T},\cdot)=(0,\infty)$ and thus extending $U(\o^{T},\cdot)$ by continuity in $0$ is enough to get an usc function
for all $\o^T \in \Omega^T$. If $\mbox{Dom} \, U(\o^{T},\cdot)=(0,\infty)$ is not true on a Borel  $\mathcal{Q}^{T}$-full measure set  then   one cannot avoid the usc assumption~: $ U(\o^{T},\cdot)$ is continuous on $\mbox{Ri}(Dom\, U(\o^{T},\cdot))=(m(\o^{T}),\infty)$  and one need to  extend $U(\o^{T},\cdot)$ by (right)-continuity in $m(\o^{T})$ which might be strictly positive. This is the reason why in the dynamic programming part we force the value function to be usc on all $\Omega^t$ by taking their closure (see Lemma \ref{concaveandusc}, \eqref{vanek2} and \eqref{amiens}).
Note that we can easily include the case where $U(\o^{T}, \cdot)$ is  non-decreasing and  concave  only for $\o^T$ in a Borel  $\mathcal{Q}^{T}$-full measure set. We introduce the following notations.
\end{remark}
\begin{definition}
\label{defGeneral}
Fix some $x\geq 0$.
For $P \in \mathfrak{P}(\O^{T})$ fixed, we denote by $\Phi(x,P)$  the set of all strategies $\phi \in \Phi$ such that $V_{T}^{x,\phi}(\cdot)\geq 0$ $P$-a.s.  and by $\Phi(x,U,P)$ the set of all strategies $\phi \in \Phi(x,P)$ such that either $E_{P}U^{+}(\cdot,V_{T}^{x,\phi}(\cdot))<\infty$ or $E_{P}U^{-}(\cdot,V_{T}^{x,\phi}(\cdot))<\infty$. Then
\begin{align}
\label{phix}
\Phi(x,\mathcal{Q}^{T}):= \bigcap_{P \in \mathcal{Q}^{T}} \Phi(x,P)\; \mbox{and} \;
\Phi(x,U,\mathcal{Q}^{T}):= \bigcap_{P \in \mathcal{Q}^{T}} \Phi(x,U,P).
\end{align}
\end{definition}
Under $NA(\mathcal{Q}^{T})$, if $\phi \in \Phi(x, \mathcal{Q}^{T})$ then $P_{t}(V_{t}^{x,\phi}(\cdot)\geq 0)=1$ for all $P \in \mathcal{Q}^{t}$ and $1 \leq t \leq T$, see Lemma \ref{AOAT}. Note that in \citep[Definition of $\Hc_x$, top of p10]{Nutz}, this intertemporal budget constraint was postulated.
We now state our main concern.
\begin{definition}
\label{DefU}
	Let $x \geq 0 $, the \emph{multiple-priors portfolio problem} with initial wealth $x$  is
	\begin{align}\label{eq:OP}
		u(x) :=  \sup_{\phi \in \Phi(x,U,\mathcal{Q}^{T})} \inf_{P \in \mathcal{Q}^{T}} E_{P} U(\cdot,V^{x,\phi}_{T}(\cdot)).
	\end{align}
\end{definition}
\begin{remark}
\label{verslinfinietaudela}
We will use the convention $+\infty-\infty=+\infty$ throughout the paper. This choice is rather unnatural when studying maximisation problem. The reason for this is that we will  use \citep[Proposition 7.48 p180]{BS} (which relies on  \citep[Lemma 7.30 (4) p177]{BS}) for lower-semianalytic function where this convention is required.
\end{remark}
We now present our main result under conditions which will be detailed in section \ref{se3}.
\begin{theorem}
\label{main}
Assume that the $NA(\mathcal{Q}^{T})$ condition and  Assumptions \ref{Qanalytic}, \ref{Sass},  \ref{Sass2}, \ref{Uminus} and \ref{Uplus} hold true. Let $x\geq 0$. Then, there exists some optimal strategy $\phi^* \in \Phi(x,U, \mathcal{Q}^{T})$ such that
$$u(x) =  \inf_{P \in \mathcal{Q}^{T}} E_{P} U(\cdot,V^{x,\phi^*}_{T}(\cdot))<\infty.$$
\end{theorem}
In Theorem \ref{main2}, we will propose a fairly general set-up where Assumption \ref{Uplus} is satisfied.
\section{No-arbitrage condition characterisation}
\label{secna}
We will often use the following one-period version of the no-arbitrage condition. For $\o^{t} \in \O^{t}$ fixed we  say that $NA(\mathcal{Q}_{t+1}(\o^{t}))$ condition holds true if for all $h \in \mathbb{R}^{d}$
\begin{align} \label{NP1}
h\Delta S_{t+1}(\o^{t},\cdot) \geq 0 \; \mathcal{Q}_{t+1}(\o^{t}) \mbox{-q.s.} \Rightarrow h\Delta S_{t+1}(\o^{t},\cdot) = 0 \;\mathcal{Q}_{t+1}(\o^t)\mbox{-q.s.} \end{align}
We introduce the affine hull (denoted by $\mbox{Aff}$) of the (robust) conditional support of $\Delta S_{t+1}$.
\begin{definition}
\label{DefD}
Let  $0 \leq t \leq T-1$ be fixed, the  random set  $D^{t+1} : \Omega^{t} \twoheadrightarrow \mathbb{R}^{d}$ is defined as
$$
 D^{t+1}(\o^{t}):= \mbox{Aff} \left( \bigcap  \left\{ A \subset \mathbb{R}^{d},\; \mbox{closed}, \; P_{t+1}\left(\Delta S_{t+1}(\o^{t},.) \in A\right)=1, \; \forall \,P_{t+1} \in \mathcal{Q}_{t+1}(\o^{t}) \right\} \right).
$$
\end{definition}
A strategy $\phi  \in \Phi$ such that $\phi^{t+1}(\o^{t}) \in D^{t+1}(\o^{t})$ have nice properties, see \eqref{valaki} and Lemma \ref{rast}.
If $D^{t+1}(\o^{t}) = \mathbb{R}^d$ then,
intuitively, there are no redundant assets for all model specifications. Otherwise, for any $\mathcal{B}_{c}(\O^{t})$-measurable strategy $\phi_{t+1}$, one may
always replace $\phi_{t+1}(\o^{t},\cdot)$ by its orthogonal projection
$\phi^{\perp}_{t+1}(\o^{t},\cdot)$ on $D^{t+1}(\o^{t})$ without changing the portfolio value (see Remark \ref{proj} below and \citep[Lemma 2.6]{Nutz}).
The following lemma establishes some important properties of  $D^{t+1}$.
\begin{lemma}
\label{Dmeasurability}
Let Assumptions \ref{Qanalytic} and \ref{Sass} hold true and $0 \leq t \leq T-1$ be fixed.
Then $D^{t+1}$ is a non-empty, closed valued random set and  $\mbox{Graph}(D^{t+1}) \in \mathcal{B}_{c}(\O^t) \otimes \mathcal{B}(\mathbb{R}^{d})$.
\end{lemma}
\begin{proof} The proof uses similar arguments as  in  \citep[Theorem 14.8 p648, Ex. 14.2 p652]{rw} together with  \citep[Lemma 4.3]{BN} and is thus omitted.
\end{proof}\\
Similarly as in \citep{RS05} and \citep{JS98} (see also \citep{BCR}), we prove  a ``quantitative'' characterisation of the $NA(\mathcal{Q}^{T})$ condition.
\begin{proposition}
\label{thelemmamultiper}
Assume that the $NA(\mathcal{Q}^{T})$ condition and  Assumptions \ref{Qanalytic}, \ref{Sass} hold true. Then for all $0\leq t\leq T-1$, there exists some $\mathcal{Q}^{t}$-full measure set $\Omega^{t}_{NA} \in \mathcal{B}_{c}(\O^{t})$   such that for all $\omega^{t} \in \Omega^{t}_{NA}$, $NA(\mathcal{Q}_{t+1}(\omega^{t}))$ holds true, $D^{t+1}(\o^{t})$ is a vector space and there exists $\alpha_{t}(\omega^{t})>0$ such that for all  $h \in  D^{t+1}(\omega^{t})$  there exists $P_{h} \in \mathcal{Q}_{t+1}(\o^{t})$ satisfying
\begin{align}
\label{valaki}
P_{h}\left(\frac{h}{|h|}\Delta S_{t+1}(\omega^{t},.)<-\alpha_{t}(\omega^{t})\right)> \alpha_{t}(\omega^{t}).
\end{align}
\end{proposition}
We prove in \citep{BC17} that there is in fact an equivalence between the $NA(\mathcal{Q}^{T})$ condition and \eqref{valaki}. We also prove that $\o^{t}  \to \alpha_{t}(\o^{t})$ is $\mathcal{B}_{c}(\O^{t})$-measurable.\\

\noindent \begin{proof} Using  \citep[Theorem 4.5]{BN}, $N_{t}:=\{\o^{t} \in \Omega^{t},\; NA(\mathcal{Q}_{t+1}(\omega^{t})) \; \mbox{fails}\} \in \mathcal{B}_{c}(\O^{t})$ and $P(N_{t})=1$  for all $P \in \mathcal{Q}^{t}$. So setting $\Omega^{t}_{NA}:=\Omega^{t} \backslash{N_{t}}$,  we get that \eqref{NP1} holds true  for  all $\o^{t} \in \Omega^{t}_{NA}$.
We fix some $\o^{t} \in \Omega^{t}_{NA}$. If  $h \in D^{t+1}(\o^{t})$, we have that
\begin{align} \label{NPome2}
h\Delta S_{t+1}(\o^{t},\cdot) \geq 0 \; \mathcal{Q}_{t+1}(\o^{t})\mbox{-q.s.} \Rightarrow h = 0.
\end{align}
Indeed as $\o^{t} \in \O^{t}_{NA}$,  \eqref{NP1} together with  \citep[Lemma 2.6]{Nutz} imply that $h \in \left( D^{t+1}(\o^{t}) \right)^{\perp}$ the orthogonal space of $D^{t+1}(\o^{t})$ and $h =0$. Therefore, for all $h \in D^{t+1}(\o^{t})$, $h \neq 0$, there exists $P_{h} \in \mathcal{Q}_{t+1}(\o^{t})$   such that $P_{h}(h\Delta S_{t+1}(\o^t,\cdot)\geq 0)<1$. Using a slight modification of \citep[Lemma 3.5]{BCR}
we get that $0 \in D^{t+1}(\o^{t})$  ($i.e$ $D^{t+1}(\o^{t})$ is a vector space).   We introduce for $n\geq 1$
\begin{small}
\begin{align*}
A_{n}(\o^{t}):=\left\{ h \in  D^{t+1}(\o^{t}),\; |h|=1,\; P_{t+1}\left(h\Delta S_{t+1}(\o^t,\cdot) \leq -\frac{1}{n}\right) \leq \frac{1}{n},  \; \forall P_{t+1} \in \mathcal{Q}_{t+1}(\o^{t})\right\}
\end{align*}
\end{small}
and we define $n_{0}(\o^{t}):=\inf\{n \geq 1, A_{n}(\o^{t})=\emptyset\}$ with the convention that $\inf \emptyset=+\infty$.
If $D^{t+1}(\o^{t}) =\{0\}$, then $n_{0}(\o^{t})=1 <\infty$.  We assume now that $D^{t+1}(\o^{t}) \neq \{0\}$ and prove by
contradiction that  $n_{0}(\o^{t})<\infty$.
Suppose that $n_{0}(\o^{t})=\infty$. For all $ n\geq 1$, we get some $h_{n}(\o^{t}) \in D^{t+1}(\o^t)$ with $|h_{n}(\o^{t})|=1$ and such that for all $P_{t+1} \in \mathcal{Q}_{t+1}(\o^{t})$
$P_{t+1}\left(h_n(\o^{t})\Delta S_{t+1}(\o^t,\cdot) \leq -\frac{1}{n}\right) \leq \frac{1}{n}.$
By passing to a sub-sequence we can assume that $h_{n}(\o^{t})$ tends to  some $h^{*}(\o^{t})\in D^{t+1}(\o^t)$ with $|h^{*}(\o^{t})|=1$.  Then
$\{ h^{*}(\o^{t})\Delta S_{t+1}(\o^t,\cdot)<0 \} \subset \liminf_{n} B_{n}(\o^{t}),$ where $B_{n}(\o^{t}):= \{h_{n}(\o^{t})\Delta S_{t+1}(\o^t,\cdot) \leq  -1/n \}$. Fatou's Lemma implies that for any $P_{t+1} \in \mathcal{Q}_{t+1}(\o^{t})$
\begin{align*}
P_{t+1}\left(h^*(\o^{t})\Delta S_{t+1}(\o^t,\cdot) < 0\right) &
  \leq
 \liminf_{n} \int_{\Omega_{t+1}}
1_{B_n(\o^{t})}(\o_{t+1}) P_{t+1}(d\o_{t+1})=0.
\end{align*}
This implies that $P_{t+1}\left(h^*(\o^{t})\Delta S_{t+1}(\o^t,\cdot) \geq0\right)=1$ for all $P_{t+1} \in \mathcal{Q}_{t+1}(\o^{t})$ and  $h^{*}(\o^{t})=0$ (see \eqref{NPome2}),
which contradicts $|h^{*}(\o^{t})|=1$. Thus $n_{0}(\o^{t})<\infty$.  We set for $\o^{t} \in  \Omega^{t}_{NA}$,
$\alpha_{t}(\o^{t}):= \frac{1}{n_{0}(\o^{t})}$, $\alpha_{t} \in (0,1]$ and by definition of $A_{n_{0}(\o^{t})}(\o^{t})$, \eqref{valaki} holds true.
 \end{proof}\\
Finally, we introduce an alternative notion of no arbitrage, called strong no arbitrage.
\begin{definition}
\label{sNA}
We say that the  $sNA(\mathcal{Q}^{T})$ condition holds true if for all  $P \in \mathcal{Q}^{T}$ and $\phi  \in \Phi$,
$V_{T}^{0,\phi} \geq 0 \; {P}\mbox{-a.s.}$  implies that $V_{T}^{0,\phi}  = 0 \;P \mbox{-a.s.}$
\end{definition}
The  $sNA(\mathcal{Q}^{T})$ condition holds true if the ``classical" no-arbitrage condition in model $P$, $NA(P)$, holds true for all $P \in \mathcal{Q}^{T}$. Note that if $\mathcal{Q}^{T}=\{P\}$ then $sNA(\mathcal{Q}^{T})=NA(\mathcal{Q}^{T})=NA(P)$.  Clearly the $sNA(\mathcal{Q}^{T})$ condition is stronger than the $NA(\mathcal{Q}^{T})$  condition.

As in  \citep[Definition 3.3]{BCR}, we introduce for all $P=P_1 \otimes q_{2} \otimes \cdots \otimes q_{T} \in \mathcal{Q}^{T}$ and for all $1 \leq t \leq T-1$,
$$
 D_{P}^{t+1}(\o^{t}):= \mbox{Aff} \left(\bigcap  \left\{ A \subset \mathbb{R}^{d},\; \mbox{closed}, \; q_{t+1}\left(\Delta S_{t+1}(\o^{t},.) \in A,\o^{t}\right) =1\right\}\right).$$
The case $t=0$ is obtained by replacing $q_{t+1}(\cdot,\o^{t})$ by $P_{1}(\cdot)$.
\begin{proposition}
\label{AOAmultistrong}
Assume that the $sNA(\mathcal{Q}^{T})$ condition and Assumptions \ref{Qanalytic} and \ref{Sass} hold true and let $0\leq t \leq T-1$. Fix some $P= P_{1}  \otimes q_2 \otimes \cdots \otimes q_{T} \in \mathcal{Q}^{T}$. Then there exists $\Omega_{P}^{t} \in \mathcal{B}(\O^{t})$ with  $P_{t}(\Omega_{P}^{t})=1$ such that for all $\o^t \in {\Omega}_{P}^t$, there exists
$\alpha^{P}_t(\o^t)  \in (0,1]$ such that for all $h \in D_{P}^{t+1}(\o^t)$,
$q_{t+1}\big(h\Delta S_{t+1}(\o^t,\cdot)\leq-\alpha^{P}_t(\o^t)|h|,\o^t\big)\geq \alpha^{P}_{t}(\o^t)$.
Furthermore $\o^{t}  \to \alpha^{P}_{t}(\o^{t})$ is $\mathcal{B}(\O^{t})$-measurable.
\end{proposition}
\begin{proof}
This is a careful adaptation of  \citep[Proposition 3.7]{BCR} since $\mathcal{B}_{c}(\O^{t})$ is not a product sigma-algebra.
\end{proof}

\section{Utility maximisation problem}\label{se3}
\begin{assumption}\label{Uminus}
For all $r \in \mathbb{Q}$, $r >0$
$\sup_{P \in \mathcal{Q}^{T}} E_{P} U^{-}(\cdot,r) <+\infty.$
\end{assumption}
The proof of the following lemma follows directly from \citep[Theorem 10.1 p82]{cvx}.
\begin{lemma}
\label{Ucontinuity}
Assume that Assumption \ref{Uminus} holds true.  Then
$\Omega^{T}_{Dom}:=\{U(\cdot,r)>-\infty, \forall r \in \mathbb{Q}, \; r>0\} \in \Bc (\Omega^{T})$ is a $\mathcal{Q}^{T}$-full measure set. For all $\o^{T} \in \Omega^{T}_{Dom}$,  $\mbox{Ri}(\mbox{Dom}\, U(\o^{T},\cdot))=(0,\infty)$ and $U(\o^{T},\cdot)$  is continuous on $(0,\infty)$, right-continuous in $0$ and thus usc on $\mathbb{R}$.
\end{lemma}
\begin{remark}
\label{Uminusdet}
Assumption \ref{Uminus}, which does not appear in the mono-prior case (see \citep{BCR}), allows to work with countable supremum (see \eqref{vanek}) and to have value functions with ``good" measurability properties (see also Remark \ref{Vminrem}). We  will prove (see Proposition \ref{ToolJC}) that Assumption \ref{Uminus} is preserved through the dynamic programming procedure. Assumption \ref{Uminus} is  superfluous in the case of non-random utility function.   Indeed let
$m:=\inf\{x \in \mathbb{R},\; U(x) >-\infty\} \geq 0$ and $\overline{U}(x) = U(x+m)$. Then $\mbox{Ri}(\mbox{Dom}\, \overline{U}(\cdot))=(0,\infty)$, $\overline{U}$ satisfies Definition \ref{utilite} and if $\overline{\phi}^{*}$ is a solution of \eqref{eq:OP} for $\overline{U}$ with an initial wealth $x$, then it will be a solution of  \eqref{eq:OP} for $U$ starting from $x+m$. Assumption \ref{Uminus} is also useless in the one-period case.
\end{remark}
\begin{example}
 We propose the following example where Assumption \ref{Uminus} holds true. Assume that there exists some $x_0>0$ such that $\sup_{P \in \mathcal{Q}^{T}} E_{P}U^{-}(\cdot,x_0)<\infty$. Assume also that there exists  some functions $f_1, f_{2}: (0,1] \to (0,\infty)$  as well as some non-negative $\mathcal{B}_{c}(\O^{T})$-measurable random variable  $D$ verifying  $\sup_{P \in \mathcal{Q}^{T}} E_{P} D(\cdot)<\infty$ such that for all $\omega^{T} \in \Omega^{T}$,  $x\geq 0$,  $0 < \l \leq 1$,  $U(\omega^{T},\lambda x) \geq f_{1}(\l)   U(\omega^{T},x) -f_{2}(\l) D(\o^{T})$.
This condition is a kind of elasticity assumption around zero. It is satisfied for example by the logarithm function.   Fix some $r \in \mathbb{Q}$, $r>0$. If $r \geq x_0$, it is clear from Definition \ref{utilite}  that $\sup_{P \in \mathcal{Q}^{T}} E_{P}U^{-}(\cdot,r)<\infty$. If $r <x_0$,  we have for all $\omega^{T} \in \Omega^{T}$, $ U(\omega^{T},r) \geq f_{1}(\frac{r}{x_0})U(\omega^{T},x_0)-f_{2}(\frac{r}{x_0})D(\o^{T})$ and  $\sup_{P \in \mathcal{Q}^{T}} E_{P}U^{-}(\cdot,r)<\infty$ follows immediately.
\end{example}
The following condition (together with Assumption \ref{Uminus}) implies that if  $\phi \in \Phi(x, \mathcal{Q}^{T})$ then  $E_{P} U(\cdot,V^{x,\phi}_{T}(\cdot))$ is well defined for all $P \in \mathcal{Q}^{T}$ (see Proposition \ref{propufini}). It also allows us  to work with  auxiliary  functions which play the role of properly integrable bounds for the value functions at each step (see \eqref{Vvanek}, \eqref{reims2}, \eqref{reims3} and \eqref{dimancheR}).
\begin{assumption}\label{Uplus}
We assume that
$\sup_{P \in \mathcal{Q}^{T}} \sup_{\phi \in \Phi(1,P)} E_{P}U^+(\cdot,V_{T}^{1,\phi}(\cdot))<\infty.$
\end{assumption}

Assumption \ref{Uplus} is not easy to verify~: we  propose an application of  Theorem \ref{main} in the following  fairly general set-up where Assumption \ref{Uplus} is automatically satisfied. We introduce for all $1 \leq t \leq T$, $r>0$,\\
 $$\mathcal{W}^{r}_{t}:= \left \{ X: \Omega^{t} \to \mathbb{R}\cup \{\pm \infty\}, \; \mbox{$\mathcal{B}(\O^{t})$-measurable},\; \sup_{P \in \mathcal{Q}^{t}}E_{P} |X|^{r} <\infty \right\}\; \mbox{and}
\; \mathcal{W}_{t}:= \bigcap_{r>0} \mathcal{W}^{r}_{t}.$$
In \citep[Proposition 14]{DHP11} it is proved  that  $\mathcal{W}^{r}_{t}$ is a Banach space (up to the usual quotient identifying two random variables that are $\mathcal{Q}^{t}$-q.s. equal) for the norm $||X||:=\left(\sup_{P \in \mathcal{Q}^{t}}E_{P} |X|^{r}\right)^{\frac{1}{r}}$.  Hence, the space $\mathcal{W}_{t}$ is the ``natural" extension of the one introduced in the mono-prior classical  case (see \citep{CR14} or \citep[(16)]{BCR}).
\begin{theorem}
\label{main2}
Assume that the $sNA(\mathcal{Q}^{T})$ condition and  Assumptions \ref{Qanalytic}, \ref{Sass}, \ref{Sass2} and \ref{Uminus} hold true. Assume furthermore that $U^{+}(\cdot,1), U^{-}(\cdot,\frac{1}{4}) \in \mathcal{W}_{T}$  and that  for all $1 \leq t \leq T$, $P \in \mathcal{Q}^{t}$, $\Delta S_{t}, \frac{1}{\alpha^{P}_{t}} \in \mathcal{W}_{t}$ (recall Proposition \ref{AOAmultistrong} for the definition of $\alpha_{t}^{P}$).   Let $x\geq 0$. Then,
there exists some optimal strategy $\phi^* \in \Phi(x,U, \mathcal{Q}^{T})$ such that
$$u(x) =  \inf_{P \in \mathcal{Q}^{T}} E_{P} U(\cdot,V^{x,\phi^*}_{T}(\cdot))<\infty.$$\end{theorem}

\subsection{\textbf{One period case}}
\label{seone}
Let $(\overline{\Omega}, \Gc)$ be a measurable space, $\mathfrak{P}(\overline{\Omega})$ the set of all probability measures on $\overline{\Omega}$ defined on $\mathcal{G}$  and $\mathcal{Q}$  a non-empty convex subset of $\mathfrak{P}(\overline{\Omega})$.  Let $Y(\cdot):=\left(Y_{1}(\cdot), \cdots, Y_{d}(\cdot)\right)$ be a $\Gc$-measurable $\mathbb{R}^{d}$-valued random variable (which could represent the change of value of the price process).
\begin{assumption}
\label{Yb}
There exists some constant $0<b <\infty $ such that  $Y_{i}(\cdot) \geq -b$ for all $i=1, \cdots, d$.
\end{assumption}
Finally, as in Definition \ref{DefD}, $D \subset \mathbb{R}^d$ is
the smallest affine subspace of $\mathbb{R}^d$ containing  the support of the
distribution of $Y(\cdot)$ under $P$ for all $P \in \mathcal{Q}$.
\begin{assumption}
\label{D0}
The set $D$ contains 0 ($D$ is a non-empty vector subspace of $\mathbb{R}^d$).
\end{assumption}
The pendant of the $NA(\mathcal{Q}^{T})$  condition in the one-period model  is given by
\begin{assumption}
\label{AOAone}
There exists some constant $0<\alpha\leq 1$
such that for all  $h \in D$ there exists $P_{h} \in \mathcal{Q}$ satisfying $P_{h}( h Y(\cdot) \leq -\alpha |h|) \geq \alpha$.
\end{assumption}
\begin{remark}
\label{proj}
Let $h \in  \mathbb{R}^d$ and $h' \in  \mathbb{R}^d$  be the orthogonal
projection of $h$ on $D$. Then $h-h'\perp D$  hence
$$\{Y(\cdot)\in D\}\subset\{(h-h')Y(\cdot)=0\}=\{hY(\cdot)=h' Y(\cdot)\}.$$
By   definition of $D$ we have $P(Y(\cdot)\in D)=1$ for all  $P \in \mathcal{Q}$ and therefore $hY=h'Y$ $\mathcal{Q}$-q.s.
\end{remark}
For $x \geq 0$ and $a \geq 0$ we define
\begin{align}
\label{HxDx}
\Hc_{x}^{a} :=
\left\{h \in \mathbb{R}^{d}, \; x+ h Y\geq a \; \mathcal{Q}\mbox{-q.s.}\right\} \mbox{ and }
D_{x}:= \Hc_{x} \cap D, \mbox{ where } \Hc_{x}:= \Hc_{x}^{0}.
\end{align}
\begin{lemma}
\label{rast}
Assume that Assumption \ref{AOAone} holds true. Then for all $x \geq 0$, $D_{x} \subset B(0,\frac{x}{\alpha})$  where $B(0,\frac{x}{\alpha})=
\{h \in \mathbb{R}^{d}, \ |h| \leq \frac{x}{\alpha}\}$ and $D_{x}$  is a convex and compact subspace of $\mathbb{R}^d$ .
\end{lemma}
\begin{proof}
 For  $x \geq0$, the convexity and the closedness of  $D_x$ are clear.  Let   $h \in
{D}_x$ be fixed. Assume that $|h| > \frac{x}{\alpha}$, then from Assumption \ref{AOAone}, there exists $P_{h} \in \mathcal{Q}$ such that
$P_{h}(x+ hY(\cdot) <0) \geq P_{h}( h Y(\cdot) \leq -\alpha |h|) \geq \alpha>0$, a contradiction.  The compactness of $D_{x}$ follows immediately.
\end{proof}

\begin{assumption}
\label{samedi}
We consider a {function} $V:~\overline{\Omega} \times \mathbb{R} \rightarrow\mathbb{R}\cup \{\pm \infty\}$ such that for every $x\in \mathbb{R}$,  $V\left(\cdot,x\right):\overline{\Omega}\rightarrow\mathbb{R}  \cup \{\pm \infty\}$ is $\Gc$-measurable, for every $\omega\in \overline{\Omega}$,  $V\left(\o,\cdot\right):\mathbb{R}\rightarrow\mathbb{R}  \cup \{\pm \infty\}$ is non-decreasing, concave and usc,
and $V(\cdot,x)=-\infty$, for all $x <0$.
\end{assumption}
The reason for not excluding at this stage improper concave function is related to the multi-period case. Indeed if Assumption \ref{AOAone} is not verified, then $v$ (or $ v^\mathbb{Q}$, $\mbox{Cl}({v^\mathbb{Q}})$) might be equal to $+\infty$.  So in the multi-period part, finding a version of the value function that is proper for all $\o^{t}$ while preserving its measurability is challenging since $\O^{t}_{NA}$ (the set where Assumption  \ref{AOAone} holds true, see Proposition \ref{thelemmamultiper}) is only  universally-measurable.
So here we do not assume that  $V\left(\o,\cdot\right)$ is proper but we will prove in Theorem \ref{main1} that the associated value function is finite. We also assume that $V(\o,\cdot)$ is usc for all $\o$, see Remark \ref{uscoupas}.

\begin{assumption}
\label{vminus}For all $r  \in \mathbb{Q}$, $r>0$,
$\sup_{P \in \mathcal{Q}} E_{P}V^-\left(\cdot,r \right)<\infty$.
\end{assumption}
\begin{remark}
\label{Vminrem}
This assumption is essential to prove in Theorem \ref{main1} that \eqref{VQ} holds true as it allows to prove that $\mathbb{Q}^{d}$ is dense in $\mbox{Ri} \left(\{h \in \mathcal{H}_{x}, \; \inf_{P \in \mathcal{Q}} E V(\cdot,x+hY(\cdot)) > -\infty\}\right)$. Note that the one-period optimisation problem in  \eqref{parc} could be solved without Assumption \ref{vminus} (see Remark \ref{Uminusdet}).
\end{remark}
The following lemma is similar to Lemma \ref{Ucontinuity} (recall also
(see \citep[Lemma 7.12]{BCR}).
\begin{lemma}
\label{Vcontinuity}
Assume that Assumptions  \ref{samedi} and \ref{vminus} hold true. Then $
\Omega_{Dom}:=\{V(\cdot,r)>-\infty, \; \forall r \in \mathbb{Q}, \; r>0\} \in \Gc$ and  $\Omega_{Dom}$ is $\mathcal{Q}$-full measure set on which $\mbox{Ri}(Dom\, V(\o,\cdot))=(0,\infty)$ and thus $V(\o,\cdot)$  is continuous on $(0,\infty)$. Moreover $V(\o,\cdot)$  is right-continuous in $0$ for all $\o \in \overline{\O}$.
\end{lemma}
Our main concern in the one period case is the following optimisation problem
\begin{align}\label{parc}
		v(x) :=
		\begin{cases}
		\sup_{h \in \Hc_{x} } \inf_{P \in \mathcal{Q}} E_{P}V\left(\cdot,x + hY(\cdot)\right), \; \mbox{ if $x \geq 0$}\\
		-\infty, \; \mbox{otherwise}.
		\end{cases}
	\end{align}
We use the convention $\infty-\infty=\infty$ (recall Remark \ref{verslinfinietaudela}), but we will see in Lemma \ref{fat} that under appropriate assumptions,  $E_{P}V(\cdot,x + hY(\cdot))$ is  well-defined.  Note  also that  for $x \geq 0$ (see Remark \ref{proj})
\begin{align}\label{noam}
v(x)  =  \sup_{h \in D_x}  \inf_{P \in \mathcal{Q}}  E_{P}V(\cdot,x + hY(\cdot)).
\end{align}
We present now some integrability assumptions on  $V^{+}$ which allow to assert that there exists some optimal solution for  \eqref{parc}.
\begin{assumption}
\label{dimanche}
For every $P \in \mathcal{Q}$, $h \in \Hc_{1}$,
$E_{P}V^+(\cdot,1+ h Y(\cdot))<\infty.$
\end{assumption}
\begin{remark}
If Assumption \ref{dimanche} is not true,  \citep[Example 2.3]{Nutz}  shows that  one can find a counterexample where $v(x)<\infty$ but the supremum is not attained in \eqref{parc}. So one cannot use   the  ``natural" extension of the mono-prior approach, which should be
that there exists some $P \in \mathcal{Q}$ such that $E_{P}V^+(\cdot,1+ h Y(\cdot))<\infty$ for all $h \in \Hc_{1}$ (see \citep[Assumption 5.9]{BCR}).
\end{remark}

We define now
\begin{align}\label{parcQ}
		v^{\mathbb{Q}}(x) :=
		\begin{cases}
		\sup_{h \in \Hc_{x} \cap \mathbb{Q}^{d}} \inf_{P \in \mathcal{Q}} E_{P}V\left(\cdot,x + hY(\cdot)\right), \; \mbox{ if $x \geq 0$}\\
		-\infty, \; \mbox{otherwise}.
		\end{cases}
	\end{align}
Finally, we  introduce  the closure of $v^{\mathbb{Q}}$ denoted by  $\mbox{Cl}(v^{\mathbb{Q}})$ which is  the smallest usc  function $w: \mathbb{R} \to \mathbb{R} \cup \{ \pm \infty\}$ such that $w \geq v^{\mathbb{Q}}$.
We will show in Theorem \ref{main1} that $ v(x)=v^{\mathbb{Q}}(x)=\mbox{Cl}({v}^{\mathbb{Q}})(x)$, which allows in the multiperiod case (see \eqref{vanek}) to work with a countable supremum (for measurability issues)  and an usc value function (see Remark \ref{uscoupas}).  But first we provide two lemmata which  are stated  under Assumption \ref{samedi} only. They will be used in the multi-period part to prove that the value function is usc, concave (see \eqref{amiens} and  \eqref{amiens2}) and dominated (see \eqref{reims3}) for  all $\o^{t}$. This avoid difficult measurability issues when proving \eqref{reims} and \eqref{reims2} coming from full-measure sets which are not Borel and on which
Assumptions  \ref{D0},  \ref{AOAone}, \ref{vminus} and \ref{dimanche} hold true. This can be seen for example in the beginning of the proof of Proposition \ref{dyn3} where we need to apply Lemma \ref{concaveandusc} using only Assumption \ref{samedi}.
\begin{lemma}
\label{concaveandusc}
Assume that Assumption \ref{samedi} holds true. Then $v$, $v^\mathbb{Q}$ and $\mbox{Cl}({v^\mathbb{Q}})$ are concave and non-decreasing on $\mathbb{R}$ and $
 \mbox{Cl}(v^\mathbb{Q})(x) =\lim_{ \substack{\delta \to 0\\ \delta>0}}v^{\mathbb{Q}}(x + \delta ).$
\end{lemma}
\begin{proof}
As $V$ is non-decreasing (see Assumption \ref{samedi}),  $v$ and $v^\mathbb{Q}$ are clearly non-decreasing. The proof of the concavity of  $v$ or $v^\mathbb{Q}$ relies on a midpoint concavity argument and on Ostrowski Theorem, see \citep[p12]{WFD}. It is very similar to  \citep[Proposition 2]{RS06} or \citep[Lemma 3.5]{Nutz} and   thus omitted.   Using  \citep[Proposition 2.32 p57]{rw}, we obtain that  $\mbox{Cl}({v}^\mathbb{Q})$  is concave on $\mathbb{R}$. Then, using for example  \citep[1(7) p14]{rw}, we get that for all $x \in \mathbb{R}$,
$
\mbox{Cl}({v}^\mathbb{Q})(x)= \lim_{ \delta \to 0} \sup_ {|y-x| <\delta} v^{\mathbb{Q}}(y)=  \lim_{ \substack{\delta \to 0\\ \delta>0}}v^{\mathbb{Q}}(x + \delta )
$ and the proof is completed.
  \end{proof}\\
Let $x \geq 0$ and $P \in \mathcal{Q}$ be fixed. We introduce
$H_{x}(P):=\left\{h \in \mathbb{R}^{d}, \; x+ h Y\geq 0 \; P\mbox{-a.s.}\right\}$.
Note that $\mathcal{H}_{x} = \bigcap_{P \in \mathcal{Q}} H_{x}(P)$ (see  \eqref{HxDx}).
\begin{lemma}
\label{IsupV}
Assume that Assumption \ref{samedi} holds true. Let   $I: \overline{\O} \times \mathbb{R} \to [0,\infty]$ be a function such that for all $x \in \mathbb{R}$ and $h \in \mathbb{R}^{d}$, $I(\cdot,x+hY(\cdot))$ is $\mathcal{G}$-measurable,  $I(\o,\cdot)$ is non-decreasing and non-negative for all $\o \in \overline{\O}$ and $V \leq I$. Set
\begin{align*}
i(x):=
1_{[0,\infty)}(x)\sup_{h \in \mathbb{R}^{d}} \sup_{P \in \mathcal{Q}}  1_{H_{x}(P)}(h)E_{P} I(\cdot,x+hY(\cdot)).
\end{align*}
Then $i$ is non-decreasing, non-negative on $\mathbb{R}$  and   $\mbox{Cl} ({v^\mathbb{Q}})(x) \leq i(x+1)$ for all $x \in \mathbb{R}$.
\end{lemma}
\begin{proof}
Since  $I(\cdot,x+hY(\cdot))$ is $\mathcal{G}$-measurable for all $x \in \mathbb{R}$  and  $I \geq 0$, the integral in the definition of $i$ is well-defined (potentially equals to $+\infty$). It is clear   that $i$ is non-decreasing and non-negative on $\mathbb{R}$.
As   $V \leq I$ and $\mathcal{H}_{x} \subset H_{x}(P)$ if $P \in \mathcal{Q}$, it is clear that  $v^\mathbb{Q}(x) \leq i(x)$ for $x \geq 0$. And since $v^\mathbb{Q}(x)=-\infty < i(x)=0$ for $x<0$, $v^{\mathbb{Q}} \leq i$ on $\mathbb{R}$ (note that $v \leq i$ on $\mathbb{R}$ for the same reasons). Applying Lemma \ref{concaveandusc},  $\mbox{Cl}({v}^\mathbb{Q})(x) \leq v^\mathbb{Q}(x+1) \leq i(x+1)$ for all $x \in \mathbb{R}$.
\end{proof}\\

\begin{proposition}\label{ae1}
Assume that Assumptions \ref{samedi} and  \ref{vminus} hold true. Then there
exists some non negative $\Gc$-measurable random variable $C$ such that $ \sup_{P \in \mathcal{Q}}E_{P}(C) <\infty$ and  for all $\omega \in {\O}_{Dom}$ (see Lemma \ref{Vcontinuity}), $\lambda\geq 1$, $x \in \mathbb{R}$ we have
\begin{align}
\label{elastic}
V(\omega,\lambda x)  &\leq  2\lambda\left(V\left(\omega,x+ \frac{1}{2}\right)+ C(\omega)\right).
\end{align}
\end{proposition}
\begin{proof}
We use similar arguments as  \citep[Lemma 2]{RS06}. It is clear that \eqref{elastic} is  true if $x<0$.
We fix  $\o \in \Omega_{Dom}$, $x \geq \frac{1}{2}$ and $\l \geq 1$. Then $\mbox{Ri}(\mbox{Dom} \, V(\o,\cdot))=(0,\infty)$ (recall Lemma \ref{Vcontinuity}).
 We assume first that there exists some $x_0 \in \mbox{Dom} \, V(\o,\cdot)$ such that   $V(\o,x_0)<\infty$. Since $V(\o,\cdot)$ is usc and concave, using similar arguments as in  \citep[Corollary 7.2.1 p53]{cvx},  we get that $V(\o,\cdot)<\infty$ on $\mathbb{R}$.  Using  the fact that  $V(\o,\cdot)$ is concave and  non-decreasing we get that  (recall that $x \geq \frac{1}{2}$)
\begin{small}
 \begin{align}
\nonumber
V\left(\o, \l x\right)  & \leq    V\left(\o,x\right)  +
\frac{V\left(\o,x\right) - V\left(\o, \frac{1}{4}\right)}{x- \frac{1}{4}}(\l -1)x \leq
  V\left(\o,x\right) + 2 \left(\l - 1\right) \left(V\left(\o,x\right)+V^-\left(\o, \frac{1}{4}\right)\right) \\
\nonumber
& \leq     V\left(\o,x\right) + 2 \left(\l - \frac{1}{2}\right) \left(V\left(\o,x\right)+V^-\left(\o, \frac{1}{4}\right)\right) + V^{-}\left(\o, \frac{1}{4}\right)\\
 \label{step1}
&  \leq   2 \l \left(V\left(\o,x\right)+ V^{-}\left(\o, \frac{1}{4}\right)\right)
 \leq 2 \l \left(V\left(\o,x+ \frac{1}{2}\right)+ V^{-}\left(\o, \frac{1}{4}\right)\right).
\end{align}
\end{small}
Fix now $0 \leq x \leq  \frac{1}{2}$ and $\l \geq 1$. Using again that $V(\o,\cdot)$ is non-decreasing and the first inequality of \eqref{step1}, $
V(\o,\lambda x) \leq V\left(\o, \lambda \left(x+\frac{1}{2}\right)\right) \leq 2 \l \left(V\left(\o,x+ \frac{1}{2}\right)+ V^{-}\left(\o, \frac{1}{4}\right)\right)$,
and Proposition \ref{ae1} is proved setting $C(\o)=V^{-}\left(\o, \frac{1}{4}\right)$ (recall Assumption \ref{vminus}) when there exists some $x_0 \in \mbox{Dom} \, V(\o,\cdot)$ such that   $V(\o,x_0)<\infty$. Now, if this is not the case,  $V(\o,x)=\infty$ for all $x \in \mbox{Dom} \, V(\o,\cdot)$,  $C(\o)=V^{-}\left(\o, \frac{1}{4}\right)=0$ and \eqref{elastic}  also holds  true for all $x \geq 0$.
\end{proof}\\

\begin{lemma}
\label{fat}
Assume that  Assumptions \ref{D0}, \ref{AOAone}, \ref{samedi}, \ref{vminus} and \ref{dimanche}  hold true. Then there exists a non negative $\Gc$-measurable $L$ such that for all $P \in \mathcal{Q}$, $E_{P} (L)<\infty$ and for all $x \geq 0$ and $h\in \mathcal{H}_{x}$,
$
V^+(\cdot,x+hY(\cdot))\leq \left(4x+1\right)L(\cdot)  \; \mathcal{Q} \mbox{-q.s.}
$
\end{lemma}

\begin{proof}
The proof is  a slight adaptation of the one of \citep[Lemma 5.11]{BCR} (see also \citep[Lemma 2.8]{Nutz}) and  is thus omitted. Note that the function $L$  is the one defined in \citep[Lemma 5.11]{BCR}.
\end{proof}\\

\begin{lemma}
\label{cont}
Assume that Assumptions \ref{D0},  \ref{AOAone},  \ref{samedi}, \ref{vminus} and \ref{dimanche} hold true. Let  $\Hc$ be the set valued function that assigns to each $x \geq 0$ the set $\mathcal{H}_x$. Then $\mbox{Graph}(\Hc)=\{(x,h) \in [0,+\infty) \times \mathbb{R}^{d},\; h \in \mathcal{H}_{x}\}$ is a closed and convex subset of $\mathbb{R}\times \mathbb{R}^{d}$.
Let $\psi: \mathbb{R} \times \mathbb{R}^{d} \to \mathbb{R} \cup \{\pm \infty\}$ be defined by
\begin{align*}
\psi(x,h):= \begin{cases}
\inf_{P \in \mathcal{Q}}E_{P} V(\cdot,x +h Y(\cdot))  \; \mbox{if $(x,h) \in \mbox{Graph}(\Hc)$},\\
-\infty \; \mbox{otherwise}.
\end{cases}
\end{align*}
Then $\psi$ is  usc and concave on $\mathbb{R} \times \mathbb{R}^{d}$, $\psi<+\infty$ on  $\mbox{Graph}(\Hc)$ and  $\psi(x,0) >-\infty$ for all $x> 0$. \end{lemma}
\begin{proof}
For all $P \in \mathcal{Q}$, we define  $\psi_{P}: \mathbb{R} \times \mathbb{R}^{d} \to \mathbb{R} \cup \{\pm \infty\}$   by $\psi_{P}(x,h)=E_{P} V(\cdot,x +h Y(\cdot))$ if $(x,h) \in \mbox{Graph}(\Hc)$  and $-\infty$ otherwise. As  in  \citep[Lemma 5.12]{BCR}, $\mbox{Graph}(\Hc)$ is a closed convex subset of $\mathbb{R}\times\mathbb{R}^{d}$, $\psi_{P}$ is usc  on $\mathbb{R} \times \mathbb{R}^{d}$ and $\psi_{P} < \infty$  on $\mbox{Graph}(\Hc)$ for all $P \in \mathcal{Q}$. Furthermore the concavity of $\psi_{P}$ follows immediately from the one of $V$.
The function $\psi=\inf_{P \in \mathcal{Q}} \psi_{P}$ is then usc and concave. As $\psi_{P} <\infty$ on $\mbox{Graph}(\Hc)$ for all $P \in \mathcal{Q}$, it is  clear that $\psi<+\infty$ on $\mbox{Graph}(\Hc)$. Finally let $x >0$ be fixed and $r \in \mathbb{Q}$ be such that $r<x$, then we have $-\infty< \psi(r,0) \leq \psi(x,0)$ (see Assumptions \ref{samedi} and \ref{vminus}).
\end{proof}\\
We are now able to state the main result of this section.
\begin{theorem}
\label{main1}
Assume that Assumptions  \ref{Yb}, \ref{D0},  \ref{AOAone}, \ref{samedi}, \ref{vminus}   and \ref{dimanche} hold true.  Then for all $x \geq 0$, $v(x) < \infty$  and there exists some optimal strategy $\widehat{h} \in D_x$ such that
\begin{align}
\label{vopti}
v(x) = \inf_{P \in \mathcal{Q}}  E_{P}(V(\cdot,x + \widehat{h}Y(\cdot))).
\end{align}
Moreover $v$ is usc, concave, non-decreasing and $\mbox{Dom}  \,v=(0,\infty)$. For all $x \in \mathbb{R}$
\begin{align}
\label{VQ}
v(x)= v^{\mathbb{Q}}(x)= \mbox{Cl}(v^{\mathbb{Q}})(x).
\end{align}
\end{theorem}
\begin{proof}
Let $x  \geq  0$ be fixed. Fix some $P \in \mathcal{Q}$. Using Lemma \ref{fat} we have that $E_{P} V(\cdot,x + hY(\cdot)) \leq  E_{P} V^+(\cdot,x + hY(\cdot)) \leq \left(4x+1\right) E_{P} L(\cdot)<\infty$,
for all $h \in \Hc_x$. Thus  ${v}(x)  < \infty$. Now if $x>0$, $v(x) \geq \psi(x,0)>-\infty$ (see Lemma \ref{cont}). Using  Lemma \ref{concaveandusc}, $v$ is  concave and non-decreasing. Thus $v$ is continuous on $(0,\infty)$. \\
\noindent From Lemma \ref{cont}, $h   \to \psi(x,h)$ is usc on  $\mathbb{R}^{d}$ and thus on $D_x$ (recall that $D_{x}$ is closed and use  \citep[Lemma 7.11]{BCR}).  Since $D_x$ is compact (see Lemma \ref{rast}), recalling \eqref{noam} and applying \citep[Theorem 2.43 p44]{Hitch}, we find that there exists some  $\widehat{h} \in D_x$ such that \eqref{vopti} holds true.\\
We prove now that $v$ is usc in $0$ (the proof works as well for all $x^* \geq 0$). Let $(x_n)_{n\geq 0}$ be a sequence of non-negative numbers converging to $0$. Let $\widehat{h}_n \in D_{x_{n}}$ be  the optimal strategies associated to $x_{n}$ in \eqref{vopti}.
Let $(n_k)_{k\geq 1}$  be a subsequence such that $\limsup_{n} {v}(x_{n}) = \lim_{k} {v}(x_{n_k})$.
Using Lemma \ref{rast}, $|\widehat{h}_{n_k}| \leq x_{n_k}/\alpha \leq 1/\alpha$ for
$k$ big enough. So we can extract a subsequence (that we still denote by $(n_k)_{k\geq 1}$) such that there exists some  $\underline{h}^*$ with   $\widehat{h}_{n_k} \to \underline{h}^*$. As  $(x_{n_{k}},\hat{h}_{n_{k}})_{k\geq 1} \in \mbox{Graph}(\Hc)$ which is a closed  subset of  $\mathbb{R}\times\mathbb{R}^{d}$ (see Lemma \ref{cont}), $\underline{h}^* \in \Hc_0$.
Thus using that $\psi$ is usc, we get that
\begin{align*}
\limsup_{n} {v}(x_n)&=\lim_k \inf_{P \in \mathcal{Q}} E_{P} V(\cdot, x_{n_k}+ \widehat{h}_{n_k} Y(\cdot))= \lim_k \psi(x_{n_k},h_{n_k})\\
&\leq \psi(0,\underline{h}^*)= \inf_{P \in \mathcal{Q}} E_{P}V(\cdot, \underline{h}^* Y(\cdot)) \leq {v}(0).
\end{align*}
For $x<0$ all the equalities in \eqref{VQ} are trivial.  We prove the first equality in \eqref{VQ} for $x \geq 0$  fixed.  We start  with the case $x=0$. If $Y=0$ $\mathcal{Q}$-q.s. then  the first equality is trivial. If $Y \neq 0$ $\mathcal{Q}$-q.s., then it is clear that ${D}_{0}=\{0\}$ (recall Assumption \ref{D0}) and  the first equality  in \eqref{VQ} is true again.
We assume now that $x>0$. From Lemma \ref{cont}, $\psi_{x}: h  \to \psi(x,h)$ is concave,  $0 \in \mbox{Dom} \,\psi_{x}$. Thus $\mbox{Ri}(\mbox{Dom} \,\psi_{x}) \neq \emptyset$ (see  \citep[Theorem 6.2 p45]{cvx}) and  we can apply  Lemma \ref{supconvex}. Assume for a moment that we have proved that $\mathbb{Q}^{d}$ is dense in $\mbox{Ri}(\mbox{Dom} \, \psi_{x})$. As  $\psi_{x}$ is continuous on $\mbox{Ri}(\mbox{Dom} \, \psi_{x})$ (recall that $\psi_{x}$ is concave), we obtain that
\begin{align*}v(x)=\sup_{h \in \mathcal{H}_{x}} \psi_{x}(h)= \sup_{h \in {\textnormal{Dom}} \,\psi_{x}} \psi_{x}(h)&=\sup_{h \in \textnormal{Ri}(\textnormal{Dom} \, \psi_{x}) } \psi_{x}(h)\\
&=
\sup_{h \in \textnormal{Ri}(\textnormal{Dom} \, \psi_{x}) \cap \mathbb{Q}^{d}} \psi_{x}(h) \leq \sup_{h \in \mathcal{H}_{x} \cap \mathbb{Q}^{d}} \psi_{x}(h)=v^{\mathbb{Q}}(x),
\end{align*}
since  $\mbox{Ri}(\mbox{Dom} \, \psi_{x}) \subset \mathcal{H}_{x}$ and the first equality in \eqref{VQ} is proved.
 It remains to prove  that $\mathbb{Q}^{d}$ is dense in $\mbox{Ri}(\mbox{Dom} \, \psi_{x})$. Fix some $h \in \mbox{Ri}(\mathcal{H}_{x})$. From  Lemma \ref{rast2}, there is some $r \in \mathbb{Q}$, $r>0$ such that $h \in \mathcal{H}_{x}^{r}$. Using  Lemma \ref{cont} we obtain that $\psi_{x}(h) \geq \psi( r,0)>-\infty$ thus $h \in \mbox{Dom} \,\psi_{x}$ and   $\mbox{Ri}(\mathcal{H}_{x}) \subset \mbox{Dom} \, \psi_{x}$. Recalling that $0 \in \mbox{Dom}\, \psi_{x}$ and that $\mbox{Ri}(\mathcal{H}_{x})$ is an open set in $\mathbb{R}^{d}$ (see Lemma \ref{rast2}) we obtain that $\mbox{Aff}(\mbox{Dom} \, \psi_{x})=\mathbb{R}^{d}$. Then $\mbox{Ri} (\mbox{Dom} \, \psi_{x})$ is an open set in $\mathbb{R}^{d}$ and the fact that $\mathbb{Q}^{d}$ is dense  in $\mbox{Ri} (\mbox{Dom} \, \psi_{x})$ follows easily.\\
The second equality  in \eqref{VQ} follows immediately~: $v^\mathbb{Q}(x)=v(x)$ for all $x \geq 0$ and ${v}$ is usc on $[0,\infty)$ thus  $\mbox{Cl}(v^\mathbb{Q})(x)=v^\mathbb{Q}(x)$ for all $x \geq 0$.
\end{proof}\\

\subsection{\textbf{Multiperiod case}}
\label{secmulti}
\begin{proposition} \label{ae}
Assume that Assumption  \ref{Uminus} holds true. Then
there exists a non negative, $\mathcal{B}(\Omega^{T})$-measurable random variable $C_{T}$ such that $\sup_{P \in \mathcal{Q}^{T}} E_{P}( C_{T}) <\infty$ and for all $\omega^{T} \in \Omega_{Dom}^{T}$ (recall Lemma \ref{Ucontinuity}), $\lambda\geq 1$ and $x \in \mathbb{R}$, we have
\begin{small}
$$U(\o^{T},\lambda x)  \leq  2\lambda\left(U\left(\o^{T},x+\frac{1}{2}\right)+  C_{T}(\o^{T})\right) \; \mbox{and} \;
U^{+}(\o^{T},\lambda x)  \leq  2\lambda\left(U^{+}\left(\o^{T},x+\frac{1}{2}\right)+  C_{T}(\o^{T})\right).$$
\end{small}
\end{proposition}
\begin{proof}
This is just Proposition \ref{ae1} for $V=U$ and $\mathcal{G}=\mathcal{B}(\O^{T})$  (recall Lemma \ref{Ucontinuity}), setting $C_{T}(\cdot)=U^{-}\left(\cdot, \frac{1}{4}\right)$. The second inequality follows immediately since $C_{T}$ is non-negative.
\end{proof}\\

\begin{proposition}
\label{propufini}
Let Assumptions   \ref{Uminus} and \ref{Uplus}   hold true and fix some $x \geq 0$. Then \begin{align*}
M_{x}:= \sup_{P \in \mathcal{Q}^{T}} \sup_{\phi \in \phi(x,P)} E_{P} U^+ (\cdot, V_{T}^{x,\phi}(\cdot))<\infty.\end{align*}
Moreover,
 $\Phi(x,U,P)= \Phi(x,P)$ for all $P \in \mathcal{Q}^{T}$ and thus $\Phi(x,U,\mathcal{Q}^{T})= \Phi(x,\mathcal{Q}^{T})$.
\end{proposition}
\begin{proof}
Fix some $P \in \mathcal{Q}^{T}$.  From Assumption \ref{Uplus} we know that $ \Phi(1,P)= \Phi(1,U,P)$ and $M_{1}<\infty$.  Let  $x \geq 0$ and $\phi \in \Phi(x,P)$ be fixed. If $x \leq 1$ then $V^{x,\phi}_T \leq V^{1,\phi}_T$, so from Definition \ref{utilite} we get that $M_{x} \leq M_{1} <\infty$ and $ \Phi(x,P)= \Phi(x,U,P)$. If $x \geq 1$, from Proposition \ref{ae} we get that for all $\omega^{T} \in \Omega_{Dom}^{T}$
\begin{small}
$$
 U^{+}(\omega^{T},V^{x,\phi}_T(\omega^{T}))  =  U^{+}\left(\omega^{T}, 2x\left(\frac{1}{2}+ \sum_{t=1}^{T} \frac{\phi_{t}(\o^{t-1})}{2x}\Delta S_{t}(\o^{t})\right)\right)
 \leq   4x \left( U^{+}(\omega^{T},V^{1,\frac{\phi}{2x}}_T(\o^T)) +C_{T}(\omega^{T})\right).$$
 \end{small}
As $\frac{\phi}{2x} \in  \Phi(\frac{1}{2},P) \subset \Phi(1,P)=\Phi(1,U,P)$, we get that $ M_{x} \leq 4x \left(M_{1}+ \sup_{P \in \mathcal{Q}^{T}} E_{P} C_{T}\right)<\infty$ (see  Proposition \ref{ae}). Thus $ \Phi(x,P)= \Phi(x,U,P)$ and the last assertion follows from \eqref{phix}.
\end{proof}\\

We introduce now the dynamic programming procedure. First we set for all $t\in\left\{0,\ldots,T-1\right\}$, $\o^t \in {\O}^t$, $P \in \mathfrak{P}(\O_{t+1})$ and $x \geq 0$
\begin{align}
\label{domainp}
	H_{x}^{t+1}(\o^{t},P)&:= \left\{h \in \mathbb{R}^d,\;  x+ h \Delta S_{t+1}(\o^t, \cdot) \geq 0 \; P\mbox{-a.s.}\right\},\\
\label{domaine}
	\Hc_x^{t+1}(\o^t)&:=  \left\{h \in \mathbb{R}^d,\;  x+ h \Delta S_{t+1}(\o^t, \cdot) \geq 0 \; \mathcal{Q}_{t+1}(\o^{t})\mbox{-q.s.}\right\}, \\
\label{domaineproj} \Dc_x^{t+1}(\o^t)&:= \Hc_x^{t+1}(\o^t) \cap D^{t+1}(\o^t),
\end{align}
where  $D^{t+1}$ was introduced in Definition \ref{DefD}.  For  all $t\in\left\{0,\ldots,T-1\right\}$, $\o^t \in {\O}^t$, $P \in \mathfrak{P}(\O_{t+1})$ and $x < 0$, we set $H_{x}^{t+1}(\o^{t},P)=\Hc_x^{t+1}(\o^t)=\emptyset$.
We  introduce now  the value functions $U_t$ from $\Omega^t \times \mathbb{R} \to \mathbb{R}\cup\{\pm \infty\}$  for all $t\in\{0,\ldots,T\}$. To do that we define the closure of a random function  $F: \Omega^{t} \times \mathbb{R} \to \mathbb{R} \cup \{\pm \infty\}$.  Fix  $\o^{t} \in \O^{t}$, then  $x  \to F_{\o^{t}}(x):=F(\o^{t},x)$ is a real-valued function and its closure is denoted by $\mbox{Cl}\left(F_{\o^{t}}\right)$.  Now  $\mbox{Cl}(F):\Omega^{t} \times \mathbb{R} \to \mathbb{R} \cup \{\pm \infty\}$ is defined  by $
\mbox{Cl}(F)(\o^{t},x):=\mbox{Cl}\left(F_{\o^{t}}\right)(x)$. For $ 0 \leq t \leq T$,  we   set for all $x \in \mathbb{R}$  and  $\o^{t} \in {\Omega}^{t}$
\begin{small}
\begin{align}
\nonumber
  \mathcal{U}_T(\o^{T},x)&:=  U (\o^{T},x)1_{\Omega^T_{Dom} \times [0,\infty) \cup \Omega^T \times (-\infty,0)}(\o^{T},x)\\
 \mathcal{U}_{t}(\o^{t},x)&:=
 \begin{cases}
 \sup_{h \in \mathcal{H}^{t+1}_{x}(\o^{t}) \cap \mathbb{Q}^{d}}\inf_{P \in \mathcal{Q}_{t+1}(\o^{t})}\int_{\O_{t+1}}U_{t+1}(\o^t,\o_{t+1},x+h\Delta S_{t+1}(\o^t,\o_{t+1}))P(d\o_{t+1}),\\
 \mbox{ if $x \geq  0$ and $-\infty$, if $x<0$} \label{vanek}
  \end{cases}\\
 \label{vanek2}
 U_{t}(\o^{t},x)&:=\mbox{Cl}( \mathcal{U}_{t}) (\o^{t},x).
  \end{align}
  \end{small}
Since $\mathcal{U}_T$ is usc (recall Lemma \ref{Ucontinuity}), it is clear that
$U_T= \mathcal{U}_T$. As already mentioned for $t=0$ we drop the dependency in $\o_{0}$ and note  $ U_{0}(x)=U_{0}(\o^{0},x)$. The convention $\infty-\infty=\infty$ is used in the integral in \eqref{vanek} (recall Remark \ref{verslinfinietaudela}), where the intersection with $\mathbb{Q}^d$ is taken since measurability issues are better handled in this way, see the discussion before \citep[Lemma 3.6]{Nutz}.
We  introduce the function $I_{t}: \Omega^{t} \times \mathbb{R} \to [0,\infty]$ which allow us to remove the boundedness assumption of \cite{Nutz} and will be used for integrability issues.
We set $I_T:= U_T^{+}$, then for all $ 0 \leq t \leq T-1$ , $x \in \mathbb{R}$  and  $\o^{t} \in {\Omega}^{t}$
\begin{small}
\begin{align}
 I_{t}(\o^{t},x)&:= 1_{[0,\infty)}(x)
\sup_{h \in \mathbb{R}^{d}} \sup_{P \in \mathcal{Q}_{t+1}(\o^{t})} 1_{ {H}^{t+1}_{x}(\o^{t},P)}(h)\int_{\O_{t+1}}I_{t+1}(\o^t,\o_{t+1},x+1+h\Delta S_{t+1}(\o^t,\o_{t+1}))P(d\o_{t+1}).
\label{Vvanek}
 \end{align}
\end{small}
\begin{lemma}
\label{Dhmes}
Assume that Assumptions \ref{Qanalytic} and \ref{Sass} hold true.
Let $0 \leq t \leq T-1$ be fixed, $G$ be a fixed non-negative, real-valued, $\mathcal{B}_{c}(\Omega^{t})$-measurable random variable
and  consider the following random sets $\Hc^{t+1}: (\o^{t},x) \twoheadrightarrow \Hc^{t+1}_{x}(\o^t) $ and $\Dc_{G}^{t+1}: \o^{t} \twoheadrightarrow \Dc_{G(\o^{t})}^{t+1}(\o^t)$. They are closed valued, $\mbox{Graph}(\mathcal{H}^{t+1}) \in \mathcal{C}A(\Omega^{t} \times \mathbb{R}\times \mathbb{R}^{d})$ and $\mbox{Graph}(\mathcal{D}_{G}^{t+1}) \in  \mathcal{B}_{c}(\O^{t})\otimes \mathcal{B}(\mathbb{R}^{d})$. Moreover $(\o^{t},P,h,x) \to 1_{H^{t+1}_{x}(\o^{t},P)}(h)$ is  $\mathcal{B}(\O^{t})\otimes  \mathcal{B}(\mathfrak{P}(\O_{t+1}))\otimes \mathcal{B}(\mathbb{R}^{d})\otimes \mathcal{B}(\mathbb{R})$-measurable.
\end{lemma}
\begin{proof}
It is clear that  $\Hc^{t+1}$ and $\mathcal{D}^{t+1}_{G}$ are closed valued. Lemma \ref{LemmaAA2}  will be in force. First it allows to prove the last assertion since $ \left\{ (\o^{t},P,h,x),\; P(x+ h \Delta S_{t+1}(\o^t, \cdot) \geq 0)=1\right\}    \in \mathcal{B}(\O^{t}) \otimes \mathcal{B}(\mathfrak{P}(\O_{t+1})) \otimes \mathcal{B}(\mathbb{R}^{d}) \otimes \mathcal{B}(\mathbb{R})$. Then it shows that
\begin{align*}
\mbox{Graph}(\mathcal{H}^{t+1}) & = & \left\{(\o^{t},x,h),  \inf_{P\in \mathcal{Q}_{t+1}(\o^{t})} P\left(x+ h \Delta S_{t+1}(\o^t, \cdot) \geq 0\right)=1\right\} \in \mathcal{C}A(\Omega^{t} \times \mathbb{R}\times \mathbb{R}^{d}).
\end{align*}
Fix some $x \in \mathbb{R}$. For any integer $k \geq 1$, $r \in \mathbb{Q}$, $r > 0$ we introduce the following $\mathbb{R}^{d}$-valued random variable and random set $ \Delta{S}_{k,t+1}(\cdot):= \Delta S_{t+1}(\cdot)1_{\{|\Delta S_{t+1}(\cdot)| \leq k\}}(\cdot)$  and   $\mathcal{H}^{r,t+1}_{k,x}(\o^t) := \left\{h \in \mathbb{R}^d,\;  x+  \Delta S_{k,t+1}(\o^t, \cdot) \geq r \; \mathcal{Q}_{t+1}(\o^{t})\mbox{-q.s.}\right\}$  for all  $\o^{t} \in \O^{t}$. In the sequel, we will write $ \mathcal{H}^{t+1}_{k,x}(\o^t)$ instead of   $\mathcal{H}^{0,t+1}_{k,x}(\o^t)$. We first prove that $\mbox{Graph} \left(\mathcal{H}^{t+1}_{x}\right)\in \mathcal{B}_{c}(\O^{t})  \otimes \mathcal{B}(\mathbb{R}^{d})$ (recall \eqref{domaine}).
Since  $\mathcal{H}^{t+1}_{x}(\cdot)= \bigcap_{\substack{k \in \mathbb{N},\;k\geq 1}} \mathcal{H}^{t+1}_{k,x}(\cdot)$, it is enough to prove that $\mbox{Graph} \left( \mbox{Ri}(\mathcal{H}^{t+1}_{k,x})\right) \in \mathcal{B}_{c}(\O^{t}) \otimes \mathcal{B}(\mathbb{R}^{d})$ for any fixed  $k \geq 1$. Indeed from Lemma \ref{rast2}, for all $\o^{t} \in \O^{t}$, $\overline{\mbox{Ri}(\mathcal{H}^{t+1}_{k,x})(\o^{t})}=\mathcal{H}^{t+1}_{k,x}(\o^{t})$   and Lemma \ref{lambda} $i)$ applies. Since
$\Delta S_{k,t+1}$  is bounded, we also get for all $\o^{t} \in \O^{t}$ that
$
\mbox{Ri}(\mathcal{H}^{t+1}_{k,x})(\o^t) =  \bigcup_{\substack{r \in \mathbb{Q},\;r > 0}} \mathcal{H}^{r,t+1}_{k,x}(\o^{t})$.
Using  Lemmata \ref{LemmaAA2} and \ref{partial}  we obtain that for all $r \in \mathbb{Q}$, $r > 0$, $
 \mbox{Graph} \left( \mathcal{H}^{r,t+1}_{k,x}\right) $ and also  $\mbox{Graph} \left( \mbox{Ri}(\mathcal{H}^{t+1}_{k,x})\right)$ are coanalytic sets. Lemma \ref{lambda} $ii)$ implies that $\mbox{Graph} \left( \mbox{Ri}(\mathcal{H}^{t+1}_{k,x})\right) \in \mathcal{B}_{c}(\O^{t}) \otimes \mathcal{B}(\mathbb{R}^{d})$.  \\
Now let
$ \Hc_{G}^{t+1}: \o^{t}   \twoheadrightarrow \Hc^{t+1}_{G(\o^{t})}(\o^t)$ then   it is easy to see that
\begin{align*}
 \mbox{Graph}(\mathcal{H}_{G}^{t+1})&= \bigcap_{\substack{n \in \mathbb{N},\;n\geq 1}} \bigcup_{\substack{q \in \mathbb{Q},\;q \geq0}} \left\{(\o^{t},h) \in \O^{t} \times \mathbb{R}  \times \mathbb{R}^{d},\; q \leq G(\o^{t}) \leq q+\frac{1}{n}, h \in \mbox{Graph}\left(\mathcal{H}^{t+1}_{q+\frac{1}{n}}\right)\right\}\\
 &\in \mathcal{B}_{c}(\O^{t})\otimes \mathcal{B}(\mathbb{R}^{d}),
 \end{align*}
since $G$ is $\mathcal{B}_{c}(\O^{t})$-measurable. So using Lemma \ref{Dmeasurability} and that $\mbox{Graph}(\mathcal{D}_{G}^{t+1})=\mbox{Graph}(\mathcal{H}_{G}^{t+1}) \cap \mbox{Graph}({D}^{t+1})$, we obtain that   $\mbox{Graph}(\mathcal{D}_{G}^{t+1}) \in \mathcal{B}_{c}(\O^{t})\otimes \mathcal{B}(\mathbb{R}^{d})$, which concludes the proof.
\end{proof}\\
We introduce for all $r \in \mathbb{Q}$, $r>0$
\begin{align}
\label{Crec0}
J^{r}_T(\o^T) &:=  U_T^{-}(\o^T,r), \mbox{ for } \o^T \in {\O}^T, \\
\label{Crec}
J^{r}_{t}(\o^{t}) & := \sup_{P \in \mathcal{Q}_{t+1}(\o^{t})} \int_{\O_{t+1} } J^{r}_{t+1}(\omega^{t}, \o_{t+1})P(d\omega_{t+1}) \mbox{ for } t \in \{0,\dots,T-1\}, \o^t \in \O^{t}.
\end{align}
As usual we will write $J_{0}^{r}=J_{0}^{t}(\o^{0})$.
\begin{proposition}
\label{ToolJC}
Assume that Assumptions \ref{Qanalytic} and \ref{Uminus} hold true.
Then for any $t \in \{0,\dots,T\}$, $r \in \mathbb{Q}$, $r>0$,  the function
 $\o^{t}   \to J^{r}_{t}(\o^{t})$ is well defined, non-negative, usa and verifies  $\sup_{P \in \mathcal{Q}^{t}} E_{P} J^{r}_{t} <\infty$. Furthermore, there exists some $\mathcal{Q}^{t}$-full measure set $ \widehat{\Omega}^{t} \in \mathcal{C}A(\O^{t})$ on which   $J^{r}_{t}(\cdot)<\infty$.
 \end{proposition}
\begin{proof}
We proceed by induction on $t$. Fix some $r \in \mathbb{Q}$, $r>0$.  For $t=T$, $J^{r}_T(\cdot)=U_T^{-}(\cdot,r)$ is non negative and usa (see Definition \ref{utilite}, Lemma \ref{Ucontinuity} and \eqref{analyticset}). We have that\\ $\sup_{P \in \mathcal{Q}^{T}} E_{P} (J^{r}_T)<\infty$  by Assumption \ref{Uminus}.   Using Lemma \ref{Ucontinuity},   $ \widehat{\Omega}^{T}:= \Omega^{T}_{Dom}  \in \mathcal{B}(\O^{T}) \subset \mathcal{C}A(\O^T)$ (see \eqref{analyticset}), $P\left( \widehat{\Omega}^{T}\right)=1$ for all $P \in \mathcal{Q}^{T}$ and  $J_{T}^{r}<\infty$ on $\widehat{\Omega}^{T}$.
Assume
now that for some $ t\leq T-1$, $J^{r}_{t+1}$ is non negative, usa and that $\sup_{P \in \mathcal{Q}^{t+1}}E_{P}(J^{r}_{t+1})<\infty$. As $J^{r}_{t+1}(\cdot) \geq 0$, it is clear that  $J^{r}_{t}(\cdot) \geq 0$.
We  apply  \citep[Proposition 7.48 p180]{BS} \footnote{As we will often use similar arguments in the rest of the paper, we provide some details at this stage.} with $X=\O^{t} \times \mathfrak{P}(\O_{t+1})$,  $Y=\O_{t+1}$, $f(\o^{t},P,\o_{t+1})=J^{r}_{t+1}(\o^{t},\o_{t+1})$ and $q(d\o_{t+1}| \o^{t},P)=P(d\o_{t+1})$.  Indeed $f$ is usa (see \citep[Proposition 7.38 p165]{BS}) ,  $(\o^{t},P) \to P(d\o_{t+1}) \in  \mathfrak{P}(\O_{t+1})$ is a $\mathcal{B}(\O^{t}) \otimes \mathcal{B}(\mathfrak{P}(\O_{t+1}))$-measurable stochastic kernel. So we get that  $j^{r}_{t}:(\omega^{t},P) \to \int_{\Omega_{t+1}} J^{r}_{t+1}(\o^{t},\o_{t+1})P(d\o_{t+1}) $ is usa.  As Assumption \ref{Qanalytic} holds true ($\mbox{Proj}_{\O^{t}} \left(\mbox{Graph} (\mathcal{Q}_{t+1})\right)=\O^{t}$),  \citep[Proposition 7.47 p179]{BS} applies and $\o^{t}  \to  \sup_{P \in \mathcal{Q}_{t+1}(\o^{t})} j^{r}_{t}(\o^{t},P) =J^{r}_{t}(\omega^{t})$ is usa.
We  set  $\Omega_{r}^{t} :=\{ \o^{t} \in \O^{t},\; J^{r}_{t}(\o^t)<\infty\}$, then  ${\Omega}_{r}^{t}  =\bigcup_{n \geq 1}  \{ \o^{t} \in \O^{t},\; J^{r}_{t}(\o^t) \leq n \} \in \mathcal{C}A(\Omega^{t})$.
Fix some $\varepsilon>0$. From  \citep[Proposition 7.50 p184]{BS} (recall Assumption  \ref{Qanalytic}), there exists some analytically-measurable  $p_{\varepsilon}: \o^{t} \to \mathfrak{P}(\O_{t+1})$ ($p_{\varepsilon} \in  \mathcal{S}K_{t+1}$), such that  $p_{\varepsilon}(\cdot,\o^{t}) \in \mathcal{Q}_{t+1}(\o^{t})$ for all $\o^{t} \in \O^{t}$ and
\begin{align}
\label{eqjt}
j^{r}_{t}(\o^{t},p_{\varepsilon})=\int_{\Omega_{t+1}} J^{r}_{t+1}(\o^{t},\o_{t+1})p_{\varepsilon}(d\o_{t+1},\o^{t}) \geq \begin{cases} J^{r}_{t}(\o^t)- \varepsilon \mbox{ if $\o^t \in {\Omega}_{r}^{t}$}\\
\frac{1}{\varepsilon} \; \mbox{otherwise}.
\end{cases}
\end{align} Assume that $\Omega_{r}^{t}$ is not a $\mathcal{Q}^{t}$-full measure set. Then there exists some $P^{*} \in \mathcal{Q}^{t}$ such that $P^{*}({\Omega}_{r}^{t} )<1$.
Set $P^{*}_{\varepsilon}:=P^{*}\otimes p_{\varepsilon} $ then $P^{*}_{\varepsilon} \in \mathcal{Q}^{t+1}$ (see \eqref{Qstar}) and we have that
\begin{align*}
\sup_{P \in \mathcal{Q}^{t+1}} E_{P} J^{r}_{t+1} \geq E_{P^{*}_{\varepsilon}} J^{r}_{t+1}
& \geq \frac{1}{\varepsilon}  (1-P^{*}({\Omega}_{r}^{t} )) - {\varepsilon}P^{*}({\Omega}_{r}^{t} ) .
\end{align*}
As the previous inequality holds true for all $\varepsilon>0$, letting $\varepsilon$ go to $0$ we  obtain that $\sup_{P \in \mathcal{Q}^{t+1}}E_{P}(J^{r}_{t+1})=+\infty$ :  a contradiction and  $\Omega_{r}^{t}$ is a $\mathcal{Q}^{t}$-full measure set.
Now,  for all  $P \in \mathcal{Q}^{t}$, we set $P_{\varepsilon}=P\otimes p_{\varepsilon} \in \mathcal{Q}^{t+1}$ (see \eqref{Qstar}).  Then, using  \eqref{eqjt} we get that
 $$E_{P} J^{r}_{t} -\varepsilon = E_{P} 1_{\Omega^{t}_r} J^{t}_{t} -\varepsilon \leq E_{P_{\varepsilon}} J^{r}_{t+1} \leq  \sup_{P \in \mathcal{Q}^{t+1}}E_{P}(J^{r}_{t+1}).$$
Again, as this is true for all $\varepsilon>0$ and all $P \in \mathcal{Q}^{t}$ we obtain that
$ \sup_{P \in \mathcal{Q}^{t}}E_{P}(J^{r}_{t}) \leq \sup_{P \in \mathcal{Q}^{t+1}}E_{P}(J^{r}_{t+1})<\infty.$
Finally we set
$ \widehat{\Omega}^{t}= \bigcap_{r \in \mathbb{Q},\; r>0} \Omega^{t}_{r}.$
It is clear that $ \widehat{\Omega}^{t}   \in \mathcal{C}A(\Omega^{t}) $ is a $\mathcal{Q}^{t}$-full measure set and that $J_{t}^{r} (\cdot)<\infty$ on $\widehat{\Omega}^{t}$ for all $r \in \mathbb{Q}$, $r>0$. \end{proof}\\

Let $1\leq t\leq T$ be fixed.  We introduce the following notation: for any $\mathcal{B}_{c}(\O^{t-1})$-measurable random variable $G$ and any $P \in \mathcal{Q}^{t}$, $\phi_{t}(G,P)$ is the set of all  $\mathcal{B}_{c}(\O^{t-1})$-measurable random variable $\xi$ (one-step strategy),  such that $G(\cdot) + \xi \Delta S_{t}(\cdot) \geq 0$ $P$-a.s. Propositions \ref{dyn1} to \ref{dyn3}  solve the dynamic programming procedure and hold true under the following set of conditions.
 \begin{align}
\label{amiens}
& \forall \,\o^{t} \in \O^{t},\; U_{t} \left(\o^{t},\cdot\right)  :   \mathbb{R} \rightarrow\mathbb{R} \cup \{\pm \infty\} \mbox{ is non-decreasing, usc and concave on $\mathbb{R}$,}\\
\label{amiens2}
&\forall \,\o^{t} \in \O^{t},\; I_{t} \left(\o^{t},\cdot\right)  :   \mathbb{R} \rightarrow\mathbb{R}  \cup \{+ \infty\} \mbox{ is non-decreasing and non-negative on $\mathbb{R}$},\\
\label{reims}
&U_{t}  \in \mathcal{L}SA(\Omega^{t}\times\mathbb{R}),\\
\label{reims2}
&I_{t} \in  \mathcal{U}SA(\Omega^{t}\times\mathbb{R}),\\
\label{reims3}
&\mbox{$U_{t} \left(\o^{t} ,x\right)  \leq I_{t}(\o^{t},x+1)$ for all $(\o^{t},x) \in \Omega^{t} \times \mathbb{R}$},\\
\label{dimancheR}
 &\sup_{P \in \mathcal{Q}^{t}} \sup_{\xi \in \phi_{t}(G,P)} \int_{\O^{t}}I_{t}(\o^{t},G(\o^{t-1})+ \xi(\o^{t-1})\Delta S_{t}(\o^{t}))P(d\o^{t}) < \infty, \\
\nonumber &\mbox{for any  $G:=x+ \sum_{s=1}^{t-1} \phi_s \Delta S_s$,  where $x \geq 0$, $\left(\phi_s \right)_{1 \leq s \leq t-1}$ is universally-predictible}, \\
\label{bordeaux}
&U_{t}(\o^{t}, r)  \geq  -J^r_{t}(\o^{t})\; \mbox{ for all $\o^{t} \in \O^{t}$, all $r \in \mathbb{Q}$, $r>0$}.
\end{align}
\begin{proposition}\label{dyn1}
Let $0\leq t\leq T-1$ be fixed. Assume that the $NA(\mathcal{Q}^{T})$ condition, that Assumptions \ref{Qanalytic}, \ref{Sass}, \ref{Sass2} hold true and that
  \eqref{amiens}, \eqref{amiens2}, \eqref{reims}, \eqref{reims2}, \eqref{reims3}, \eqref{dimancheR} and \eqref{bordeaux} hold true at stage $t+1$.
Then there exists some $\mathcal{Q}^{t}$-full measure set $\widetilde{\Omega}^{t} \in  \mathcal{B}_{c}(\O^{t})$ such that for all $\o^{t} \in \widetilde{\Omega}^{t}$  the function $(\omega_{t+1},x) \to U_{t+1}(\o^{t},\o_{t+1},x)$ satisfies the assumptions of Theorem \ref{main1} (or Lemmata \ref{fat} and \ref{cont}) with $\overline{\Omega}=\Omega_{t+1}$, $\mathcal{G}=\mathcal{B}_{c}(\O_{t+1})$, $\mathcal{Q}=\mathcal{Q}_{t+1}(\o^{t})$,  $Y(\cdot)=\Delta S_{t+1}(\o^t,\cdot)$, $V(\cdot,\cdot)=U_{t+1}(\omega^{t},\cdot,\cdot)$ where $V$ is defined on $\Omega_{t+1} \times \mathbb{R}$ (shortly called context $t+1$ from now).
\end{proposition}
Note  that under the assumptions of Proposition \ref{dyn1}, for  all  $\o^{t} \in\widetilde{\Omega}^{t}$   and $x \geq 0$ we have that (see \eqref{VQ}, \eqref{vanek} and  \eqref{vanek2})
\begin{small}
\begin{align}
\label{Clot}
\nonumber U_{t}(\o^t,x)&= \mathcal{U}_{t}(\o^t,x)\\
 &=\sup_{h \in \Hc_x^{t+1}(\o^t)} \inf_{P \in \mathcal{Q}_{t+1}(\o^{t})} \int_{\O_{t+1} }{U}_{t+1}(\o^t, \o_{t+1}, x+ h \Delta S_{t+1}(\o^t, \o_{t+1}))P(d\o_{t+1}).
\end{align}
\end{small}
\begin{proof}
To prove the proposition we will review one by one the assumptions needed to apply Theorem \ref{main1}  in the context  $t+1$.
First  from Assumption \ref{Sass2} for  $\o^{t} \in \O^{t}$ fixed we have  that $Y_i(\cdot)= \Delta S_{t+1}^{i}(\o^{t},\cdot)  \geq -b:= -\max(1+s+S_{t}^{i}(\o^{t}), i \in \{1,\dots, d\})$ and $0 < b < \infty$:   Assumption \ref{Yb} holds true.
From \eqref{amiens} at $t+1$  for all $\o^{t} \in  \Omega^{t}$ and   $\o_{t+1}\in \Omega_{t+1}$,  $U_{t+1}(\o^{t},\o_{t+1},\cdot)$ is  non-decreasing, usc and concave on $\mathbb{R}$. From \eqref{reims} at $t+1$,  $U_{t+1}$ is $\mathcal{B}_{c}(\O^{t+1} \times \mathbb{R})$-measurable. Fix some $x \in \mathbb{R}$ and $\o^{t} \in \O^{t}$, then  $\o_{t+1}  \to U_{t+1}(\o^{t},\o_{t+1},x)$ is $\mathcal{B}_{c}(\O_{t+1})$-measurable, see  \citep[Lemma 7.29 p174]{BS}. Thus Assumption \ref{samedi} is satisfied in the context $t+1$. \\
We  now prove  the assumptions that are verified for $\o^{t}$  in some well chosen $\mathcal{Q}^{t}$-full measure set. First from  Proposition \ref{thelemmamultiper},  for all $\o^{t} \in \Omega^{t}_{NA}$, Assumptions \ref{D0} and \ref{AOAone}  hold true in the context $t+1$. Fix $\o^{t} \in  \widehat{\Omega}^{t}$ and some $r \in \mathbb{Q}$, $r>0$. Using  \eqref{bordeaux} at $t+1$ and Proposition \ref{ToolJC}, we get that
$$\sup_{P \in \mathcal{Q}_{t+1}(\o^{t})} \int_{\Omega_{t+1}} U^{-}_{t+1}(\o^{t},\o_{t+1}, r)  P(d\o^{t}) \leq \sup_{P \in \mathcal{Q}_{t+1}(\o^{t})} \int_{\Omega_{t+1}} J^{r}_{t+1}(\o^{t},\o_{t+1}) P(d\o^{t})=J^{r}_{t}(\o^{t}) <\infty,$$
and  Assumption \ref{vminus} in context $t+1$ is verified for all $\o^{t} \in   \widehat{\Omega}^{t}$.
We finish with Assumption  \ref{dimanche} in context $t+1$ whose proof is more involved. We want to show that for $\o^{t}$ in some $\mathcal{Q}^{t}$-full measure set to be determined, for all $h \in \mathcal{H}_{1}^{t+1}(\o^{t})$ and $P \in \mathcal{Q}_{t+1}(\o^{t})$ we have that  \begin{align}
\label{aim}
\int_{\Omega_{t+1}} U^{+}_{t+1}(\o^{t},\o_{t+1},1 + h \Delta S_{t+1}(\o^{t},\o_{t+1})) P(d\o_{t+1}) <\infty.
\end{align}
Let $i_{t}(\o^{t},h,P) = \int_{\Omega_{t+1}} I_{t+1}(\o^{t},\o_{t+1},2 + h \Delta S_{t+1}(\o^{t},\o_{t+1})) P(d\o_{t+1})$ and \\
$I^t(\o^{t}):=\left\{ (h,P) \in \mathbb{R}^{d} \times  \mathcal{Q}_{t+1}(\o^{t}) ,\;P\left(1+ h \Delta S_{t+1}(\o^t, \cdot) \geq 0\right) = 1,\;i_{t}(\o^{t},h,P)=\infty\right\}$. Fix some $\o^t \in \Omega^t$, then
using \eqref{amiens2} and \eqref{reims3} at $t+1$ we have that if $h \in \mathcal{H}_{1}^{t+1}(\o^{t})$ and $P \in \mathcal{Q}_{t+1}(\o^{t})$ are such that
\eqref{aim} does not hold true then $(h,P) \in I^t(\o^{t})$. Thus \eqref{aim} holds true for all $h \in \mathcal{H}_{1}^{t+1}(\o^{t})$ and $P \in \mathcal{Q}_{t+1}(\o^{t})$ if  $\o^{t} \in \{ I^t=\emptyset\}$ and if this set is of $\mathcal{Q}^{t}$-full measure, Assumption  \ref{dimanche} in context $t+1$ is proved.
We first prove that $\mbox{Graph}({I}^t) \in \mathcal{A}(\O^{t}\times \mathbb{R}^{d}  \times \mathfrak{P}(\O_{t+1}) )$.
From  \eqref{reims2} at $t+1$,  Assumption \ref{Sass} and   \citep[Lemma 7.30 (3) p178]{BS}, $(\o^{t},h,\o_{t+1}) \to I_{t+1}(\o^{t},\o_{t+1},2 + h \Delta S_{t+1}(\o^{t},\o_{t+1}))$ is usa. Then using \citep[Proposition 7.48 p180]{BS} (which can be used  with similar arguments as in the proof of Proposition \ref{ToolJC}), we get that  $i_{t}$ is usa.
It  follows that
$$i_{t}^{-1}(\{\infty\})= \bigcap_{n \geq 1} \left\{(\o^{t},h,P),\; i_{t}(\o^{t},h,P) > n \right\} \in \mathcal{A}(\O^{t} \times \mathbb{R}^{d}\times  \mathfrak{P}(\O_{t+1})).$$
Now using Assumption \ref{Qanalytic} together with   Lemma  \ref{LemmaAA2}
we get that
$$\left\{(\o^{t},h,P),\;  P \in \mathcal{Q}_{t+1}(\o^{t}),\; P\left(1+ h \Delta S_{t+1}(\o^t, \cdot) \geq 0\right) = 1 \right\}  \in \mathcal{A}(\O^{t} \times \mathbb{R}^{d}\times  \mathfrak{P}(\O_{t+1}))$$
and  the fact that $\mbox{Graph}({I}^t)$  and $Proj_{\O^{t}} \left(\mbox{Graph}({I}^t) \right)=\{I^t \neq \emptyset\} $ are analytic sets (recall  \citep[Proposition 7.39 p165]{BS}) follows immediately.
Applying  the Jankov-von Neumann Projection Theorem \citep[Proposition 7.49 p182]{BS}, we obtain that there exists some analytically-measurable and therefore $ \mathcal{B}_{c}(\O^{t})$-measurable function $ \o^{t} \in \{I^t \neq \emptyset\} \to (h^{*}(\o^{t}),p^{*}(\cdot, \o^t) ) \in \mathbb{R}^{d} \times   \mathfrak{P}(\O_{t+1}) $ such that for all $\o^{t} \in \{I^t \neq \emptyset\} $, $(h^{*}(\o^{t}),p^{*}(\cdot,\o^t)) \in I^t(\o^{t})$.  We may and will extend  $h^{*}$ and $p^{*}$ on all $\Omega^{t}$ so that  $h^{*}$ and $p^*$ remain $ \mathcal{B}_{c}(\O^{t})$-measurable.\\
We prove now  by contradiction that $\{ I^t=\emptyset\}$ is a $\mathcal{Q}^{t}$-full measure set. Assume that there exists some $\widetilde{P} \in \mathcal{Q}^{t}$ such that $\widetilde{P}( \{I^t \neq \emptyset\}) >0$ and set  $\widetilde{P}^{*}=\widetilde{P} \otimes p^{*}$. Since $p^* \in SK_{t+1}$ and $p^{*}(\cdot,\o^{t}) \in \mathcal{Q}_{t+1}(\o^{t})$ for all $\o^{t} \in \O^{t}$,  $\widetilde{P}^{*} \in \mathcal{Q}^{t+1}$ (see \eqref{Qstar}). It is also clear that  $\widetilde{P}^{*}\left(2+h^{*}(\cdot) \Delta S_{t+1}(\cdot) \geq 0 \right)=1$.  Now for all $\o^{t} \in \{I^t \neq \emptyset\}$, we have that
$i_{t}(\o^{t},h^{*}(\o^{t}),p^{*}(\cdot,\o^{t}))=\infty$ and thus
\begin{align*}
\int_{\Omega^{t+1}}& I_{t+1}\left(\o^{t+1},2 + h^{*}(\o^{t}) \Delta S_{t+1}(\o^{t+1})\right) \widetilde{P}^{*} (d\o^{t+1}) \geq  \int_{ \{I^t \neq \emptyset\}} (+\infty)  \widetilde{P} (d\o^{t})= +\infty
\end{align*}
a contradiction with \eqref{dimancheR} at $t+1$.

We can now define
$ \widetilde{\Omega}^{t}:={\{I^t = \emptyset\}} \cap  \widehat{\Omega}^{t} \cap \Omega^{t}_{NA} \subset  \widehat{\Omega}^{t}.$
 It is clear, recalling  Propositions \ref{thelemmamultiper} and \ref{ToolJC},  that $\widetilde{\Omega}^{t} \in  \mathcal{B}_{c}(\Omega^{t})$ is a $\mathcal{Q}^{t}$-full measure set and the proof is complete.
\end{proof}\\
The next proposition enables us to initialize the induction procedure that will be carried on in the proof of the main theorem.
\begin{proposition}\label{dyn2}
 Assume  that the $NA(\mathcal{Q}^{T})$ condition, Assumptions \ref{Uminus} and \ref{Uplus}  hold true. Then  \eqref{amiens}, \eqref{amiens2}, \eqref{reims}, \eqref{reims2}, \eqref{reims3}, \eqref{dimancheR} and   \eqref{bordeaux} hold true  for $t=T$.
 \end{proposition}
\begin{proof}
As
$U_T = U 1_{\Omega^T_{Dom} \times [0,\infty) \cup \Omega^T \times (-\infty,0)}$ and $I_{T}=U_T^{+}$ , using Definition \ref{utilite}, \eqref{amiens2},  \eqref{reims3}  and  \eqref{bordeaux} (recall \eqref{Crec0}) for $t=T$ are true.
For all $\o^{T} \in \O^{T}$, $U_T(\o^{T},\cdot)$ is also right-continuous and usc (see Lemma \ref{Ucontinuity}), thus \eqref{amiens} also holds true. Moreover  $U_T(\cdot,x)$ is $\mathcal{B}(\O^{T})$-measurable for all $x \in \mathbb{R}$, thus $U_T$ is $\mathcal{B}(\O^{T})\otimes \mathcal{B}(\mathbb{R})$-measurable (see \citep[Lemma 7.16]{BCR}) and  \eqref{reims} and \eqref{reims2} hold true for $t=T$.
It remains to  prove that \eqref{dimancheR} is true for $t=T$. Let $G:=x+ \sum_{t=1}^{T-1} \phi_t \Delta S_t$ where $x \geq 0$ and $\left(\phi_s \right)_{1 \leq s \leq T-1}$ is universally-predictable. Fix some $P \in \mathcal{Q}^{T}$ and $\xi \in \phi_{T}(G,P)$.
Let $(\phi^{\xi}_{i})_{1\leq i\leq T}  \in \Phi$ be defined by $\phi^{\xi}_{T}=\xi $ and $\phi^{\xi}_{s}=\phi_s$ for $1 \leq s \leq T-1$ then $V_{T}^{x,\phi^{\xi}}= G + \xi \Delta S_{T}$, $\phi^{\xi} \in \Phi(x,P)$, {$\int_{\O^{T}} I_{T}\left(\o^{T}, G(\o^{T-1}) + \xi(\o^{T-1}) \Delta S_{T}(\o^{T})\right)P(d\o^{T})= E_{P} U^{+} (\cdot, V_{T}^{x, \phi^{\xi}}(\cdot))$} and \eqref{dimancheR} follows from  Proposition \ref{propufini}.
 \end{proof}\\

The next proposition proves the induction step.
\begin{proposition}\label{dyn3}
Let $0\leq t\leq T-1$ be fixed.  Assume that the $NA(\mathcal{Q}^{T})$ condition  holds true as well as Assumptions \ref{Qanalytic},  \ref{Sass}, \ref{Sass2}  and  \eqref{amiens}, \eqref{amiens2}, \eqref{reims}, \eqref{reims2}, \eqref{reims3},  \eqref{dimancheR} and  \eqref{bordeaux} at $t+1$. Then  \eqref{amiens},  \eqref{amiens2}, \eqref{reims}, \eqref{reims2}, \eqref{reims3}, \eqref{dimancheR} and \eqref{bordeaux} are true for $t$. \\
 Moreover for all $X=x+\sum_{s=1}^{t} \phi_s \Delta S_s $, where $x \geq 0$, $\left(\phi_s \right)_{1 \leq s \leq t}$ is universally-predictable and  $\{X\geq 0\}$ is  $\mathcal{Q}^{t}$-full measure set, there exists some $\mathcal{Q}^{t}$-full measure set ${\O}^{t}_{X} \in \mathcal{B}_{c}(\O^{t})$, such that ${\O}^{t}_{X} \subset \widetilde{\O}^{t}$ (see Proposition \ref{dyn1} for the definition of $\widetilde{\O}^{t}$)  and some $\mathcal{B}_{c}(\O^{t})$-measurable random variable $\widehat{h}^X_{t+1}$  such that for all $\o^{t} \in {\O}^{t}_{X}$, $\widehat{h}^X_{t+1}(\o^{t}) \in \Dc_{X(\o^{t})}^{t+1}(\o^t)$ and
 \begin{small}
 \begin{align}
\label{noisette}
U_{t}(\omega^t,X(\omega^t)) & = \inf_{P \in \mathcal{Q}_{t+1}(\o^{t})}\int_{\O_{t+1} }U_{t+1}(\o^t, \o_{t+1}, X(\o^t)+ \widehat{h}^X_{t+1}(\o^t) \Delta S_{t+1}(\o^t, \o_{t+1}))P(d\o_{t+1}).
\end{align}
\end{small}
 \end{proposition}
\begin{proof}
First we prove that \eqref{amiens} is true at $t$. We fix some $\o^{t} \in \O^{t}$. From \eqref{amiens} at $t+1$,   the function $U_{t+1}(\o^{t},\o_{t+1},\cdot)$ is  usc, concave and non-decreasing on $\mathbb{R}$ for all  $\o_{t+1} \in \O_{t+1}$. From \eqref{vanek} and \eqref{vanek2}, $U_{t+1}(\o^{t},\o_{t+1},x)=-\infty$ for all $x<0$ and    $\o_{t+1} \in \O_{t+1}$. Then using  \eqref{reims} at $t+1$ and Lemma \ref{partial}, we find that $U_{t+1}(\o^{t},\cdot,x)$ is $\mathcal{B}_{c}(\O_{t+1})$-measurable for all $x \in \mathbb{R}$.
Hence, Assumption \ref{samedi} of Lemma \ref{concaveandusc}  holds true in the context $t+1$ and we obtain that
 $ x \to U_{t}(\o^{t},x)=\mbox{Cl}( \mathcal{U}_{t}) (\o^{t},x)$  (see \eqref{vanek} and \eqref{vanek2}) is usc, concave and non-decreasing. As this is true for all $\o^{t} \in \O^{t}$, \eqref{amiens} at $t$ is proved.  Note that we also obtain that $x \to  \mathcal{U}_{t} (\o^{t},x)$ is non decreasing for all $\o^t \in \Omega^t$.
Now  we prove \eqref{reims} at ${t}$. Since integrals might not always be well defined we need to be a bit cautious.
We introduce first  $u_{t}$ and $\widehat{u}_{t} : \Omega^{t} \times \mathbb{R}^{d}\times [0,\infty) \times \mathfrak{P}(\O_{t+1}) \to \mathbb{R} \cup \{\pm \infty\}$
\begin{align}
\nonumber
 u_{t}(\o^{t},h,x,P)&= \int_{\Omega_{t+1}} U_{t+1}(\o^{t},\o_{t+1}, x+h\Delta S_{t+1}(\omega^{t},\o_{t+1})) P(d\o_{t+1}) \\
 \label{hatu}
\widehat{u}_{t}(\o^{t},h,x,P) & =1_{\mathcal{H}^{t+1}_{x}(\o^{t})} (h) {u}_{t} (\o^{t},h,x,P) + (-\infty) 1_{\mathbb{R}^{d} \backslash{\mathcal{H}^{t+1}_{x}(\o^{t})}} (h).
\end{align}
As  $U_{t+1}$ is lsa (see \eqref{reims} at $t+1$) and Assumption \ref{Sass} holds true,  \citep[Lemma 7.30 (3) p177]{BS} implies that $(\o^{t},\o_{t+1},h,x)  \to   U_{t+1}(\o^{t},\o_{t+1}, x+h\Delta S_{t+1}(\omega^{t},\o_{t+1}))$ is lsa. So \citep[Proposition 7.48 p180]{BS} (recall the convention $\infty-\infty=\infty$, see Remark \ref{verslinfinietaudela}) shows that
 $u_{t}$ is lsa.   Fix some $c \in \mathbb{R}$ and set
$\widehat{C}:= \widehat{u}_{t}^{-1}((-\infty, c))$, ${C}:= {u}_{t}^{-1}((-\infty, c))$, $
A:=   \left\{ (\o^{t},h,x),\; h \in \mathcal{H}^{t+1}_{x}(\o^{t})\right\} \times  \mathfrak{P}(\O_{t+1})$ and
$A^{c} :=  \left\{ (\o^{t},h,x),\; h \notin \mathcal{H}^{t+1}_{x}(\o^{t})\right\} \times  \mathfrak{P}(\O_{t+1})$,
then
$\widehat{C}= \left( C \cap  A\right ) \cup  A^{c} = C \cup  A^{c}.$
As $u_{t}$ is lsa, $C$ is an analytic set. Lemma \ref{Dhmes} implies that
$
A^c= \{(\o^{t},h,x), \; (\o^{t},x,h)
\notin  \mbox{Graph}( \mathcal{H}^{t+1})\} \times  \mathfrak{P}(\O_{t+1})$, and thus $\widehat{C}$, are analytic sets and $\widehat{u}_{t}$ is lsa.
Using Assumption \ref{Qanalytic} and  \citep[Proposition 7.47 p179]{BS}, we get that
 \begin{align}
\label{widetildeu}
\widetilde{u}_{t}:(\o^{t},h,x)  \to  \inf_{P \in \mathcal{Q}_{t+1}(\o^{t})} \widehat{u}_{t} (\o^{t},h,x,P) \in \mathcal{L}SA(\Omega^{t}\times\mathbb{R}^{d}\times\mathbb{R}).
\end{align}
Then \citep[Lemma 7.30 (2) p178]{BS} implies that
$
\widetilde{\mathcal{U}}_{t}: (\o^{t},x)  \to \sup_{h \in \mathbb{Q}^{d}} \widetilde{u}_{t}(\o^{t},h,x)$ is lsa and since $\widetilde{\mathcal{U}}_{t}=\mathcal{U}_{t}$ on $\O^{t} \times [0,\infty)$, it follows that $\mathcal{U}_{t}$ is lsa.
We have already seen that $\o^{t} \in \O^{t}$, $ \mathcal{U}_{t}(\o^{t},\cdot)$ is non-decreasing, thus, for all $\o^{t} \in \O^{t}$ and $x \in \mathbb{R}$ we get that (recall \eqref{vanek2})
$$U_{t}(\o^{t},x)= \mbox{Cl}(\mathcal{U}_{t})(\o^{t},x)= \limsup_{y \to x} \mathcal{U}_{t}(\o^{t},y)= \lim_{n \to \infty} \mathcal{U}_{t}\left(\o^{t},x+\frac{1}{n}\right).$$
As $(\o^{t},x) \to \mathcal{U}_{t}(\o^{t},x+\frac{1}{n}) $ is lsa,  \citep[Lemma 7.30 (2) p178]{BS} implies that $U_{t}$ is also lsa.
We prove now that \eqref{reims2} holds true for $t$. We introduce $\hat{\i}_{t}: \O^{t} \times \mathbb{R}^{d} \times [0,\infty) \times \mathfrak{P}(\O_{t+1}) \to \mathbb{R} \cup \{+ \infty\}$ (recall  \eqref{domainp})
\begin{small}
\begin{align}
\label{itusa}
\hat{\i}_{t}(\o^{t}, h,x,P) =  1_{H_{x}^{t+1}(\o^{t},P)}(h) \int_{\Omega_{t+1}}  I_{t+1}(\o^{t},\o_{t+1},x+1+h \Delta S_{t+1}(\o^{t},\o_{t+1})) P(d\o_{t+1}).
\end{align}
\end{small}
Note that, using \eqref{amiens2} at $t+1$, the integral in \eqref{itusa}  is well defined (potentially infinite valued).
Using Assumption \ref{Sass}, \eqref{reims2} at $t+1$ and  \citep[Lemma 7.30 (3) p177]{BS} we find that $(\o^{t+1},h,x,P) \to I_{t+1}(\o^{t},\o_{t+1},x+1+h \Delta S_{t+1}(\o^{t},\o_{t+1}))$ is usa.  Thus \citep[Proposition 7.48 p180]{BS}  applies\footnote{As already mentioned,  \citep[Proposition 7.48 p180]{BS} relies on  \citep[Lemma 7.30 (4) p177]{BS} applied for upper-semianalytic functions where the convention $-\infty +\infty=-\infty$ needs to be used.  But here, as we deal with a non-negative function the convention is useless.} and
\begin{small}
$$(\o^{t}, h,x,P) \to \int_{\Omega_{t+1}} I_{t+1}(\o^{t},\o_{t+1},x+1+h \Delta S_{t+1}(\o^{t},\o_{t+1})) P(d\o_{t+1}) \in  \mathcal{U}SA(\O^{t}\times \mathbb{R}^{d} \times \mathbb{R} \times \mathfrak{P}(\O_{t+1})).$$
\end{small}
Lemma \ref{Dhmes} together with \citep[Lemma 7.30 (4) p177]{BS}   imply  that $\hat{\i}_{t}$ is usa. Finally as
$\{(\o^t,h,x,P), \, P \in \mathcal{Q}_{t+1}(\o^{t})\}$ is analytic (see Assumption \ref{Qanalytic}),
 \citep[Proposition 7.47 p179,  Lemma  7.30 (4) p178]{BS} applies and recalling  \eqref{Vvanek} and \eqref{itusa},  we get that \\
$I_{t}(\o^{t},x) = 1_{[0,\infty)}(x)\sup_{h \in \mathbb{R}^{d}} \sup_{P \in \mathcal{Q}_{t+1}(\o^{t})} \hat{\i}_{t}(\o^{t},h,x,P)$ is usa
and \eqref{reims2}  for $t$ is proved.\\
For later purpose, we set $\overline{\i}_{t}: \O^{t} \times \mathbb{R}^{d} \times[0,\infty) \times \mathfrak{P}(\O_{t+1}) \to \mathbb{R} \cup \{\pm\infty\}$
\begin{align}
\label{overlinevt}
 \overline{\i}_{t}(\o^{t},h,x,P):=  \hat{\i}_{t}(\o^{t},h,x,P)+ (-\infty) 1_{\mathbb{R}^{d} \backslash{H_x^{t+1}(\o^{t},P)}}(h).
 \end{align}
Using Lemma \ref{Dhmes},  $\overline{\i}_{t}$ is usa  and 
$\overline{I}_{t}(\o^{t},x) : = 1_{[0,\infty)}(x)\sup_{h \in \mathbb{R}^{d}} \sup_{P \in \mathcal{Q}_{t+1}(\o^{t})} \overline{\i}_{t}(\o^{t},h,x,P) $ is usa as before. Furthermore as  $\hat{\i}_{t} \geq 0$  we have  that $\overline{I}_{t}=I_t$.
To prove  \eqref{amiens2} and \eqref{reims3} at $t$, we apply  Lemma \ref{IsupV} to $V(\o_{t+1},x)= U_{t+1}(\o^{t},\o_{t+1},x)$, $I(\o_{t+1},x)=I_{t+1}(\o^{t},\o_{t+1},x+1)$ (recall \eqref{Vvanek}) and $\mathcal{G}=\mathcal{B}_{c}(\O_{t+1})$ for any fixed $\o^{t} \in \O^{t}$. Indeed we have already proved (see the proof of \eqref{amiens} at $t$) that  Assumption \ref{samedi} holds true for $V$.  From \eqref{amiens2} and \eqref{reims3} at $t+1$, $I(\o_{t+1}, \cdot)$ is non-decreasing and non-negative on $\mathbb{R}$ for all $\o_{t+1}$ and $V \leq I$. Finally using Assumption \ref{Sass} and  \eqref{reims2} at $t+1$ together with \citep[Lemma 7.30 p177]{BS}, we get that
$\o_{t+1} \to I_{t+1}(\o^{t},\o_{t+1},x+1+h \Delta S_{t+1}(\o^{t},\o_{t+1}))$ is $\mathcal{B}_{c}(\O_{t+1})$-measurable. \\
We prove now \eqref{bordeaux} at $t$. Fix some $r \in \mathbb{Q}$, $r>0$. We have from the definition of $U_{t}$ (see \eqref{vanek}, and \eqref{vanek2}), \eqref{bordeaux} at $t+1$  and the definition of $J^{r}_{t}$ (see \eqref{Crec})  that for all $\o^t \in \Omega^t$
\begin{align*}
U_{t}\left(\o^{t}, r\right) \geq \mathcal{U}_{t}\left(\o^{t}, r\right) & \geq   \inf_{P \in \mathcal{Q}_{t+1}(\o^{t})} \int_{\Omega_{t+1}} U_{t+1}\left(\o^{t},\o_{t+1},r\right) P(d\o_{t+1}) \\
 &  \geq   \inf_{P \in \mathcal{Q}_{t+1}(\o^{t})}\int_{\Omega_{t+1}} -J^{r}_{t+1}(\o^{t},\o_{t+1})P(d\o_{t+1})= -J^{r}_{t}(\o^{t}).
 \end{align*}
We prove now \eqref{dimancheR} at $t$. Choose  $x \geq 0$, $\left(\phi_s \right)_{1 \leq s \leq t-1}$  universally-predictable  random variables and  set $\overline{G}:=x+ \sum_{s=1}^{t-1} \phi_s \Delta S_s$. Furthermore, fix some $P \in \mathcal{Q}^{t}$,  $\xi \in \phi_{t}(\overline{G},P)$, $\varepsilon>0$ and set $G(\cdot):=\overline{G}(\cdot) + \xi(\cdot) \Delta S_{t}(\cdot)$. We  apply  \citep[Proposition 7.50 p184]{BS} to $\overline{\i}_{t}$ (see \eqref{overlinevt}) in order to obtain  $S^{\varepsilon}: (\o^{t},x)  \to (h^{\varepsilon}(\o^{t},x), p^{\varepsilon}(\cdot,\o^{t},x)) \in \mathbb{R}^{d} \times \mathfrak{P}(\O_{t+1})$ that is analytically-measurable such that  $p^{\varepsilon}(\cdot,\o^{t},x) \in \mathcal{Q}_{t+1}(\o^{t})$ for all $\o^{t} \in \O^{t}$, $x \geq 0$ and (recall that $\overline{I}_{t}=I_t$)
\begin{align}
\label{Vt2t}
\overline{\i}_{t}(\o^{t},h^{\varepsilon}(\o^{t},x),x,p^{\varepsilon}(\cdot,\o^{t},x))  \geq \begin{cases}
 \frac{1}{\varepsilon}, \;\mbox{if ${I}_{t}(\o^{t},x)=\infty$}\\
 {I}_{t}(\o^{t},x)-\varepsilon, \; \mbox{otherwise}.
 \end{cases}
\end{align}
Let   $h^{\varepsilon}_{G}(\o^{t}):=h^{\varepsilon}(\o^{t},1_{\{G \geq 0\}}(\o^{t})G(\o^{t}))$ and $p^{\varepsilon}_{G}(\cdot,\o^{t}):=p^{\varepsilon}(\cdot,\o^{t},1_{\{G \geq 0\}}(\o^{t})G(\o^{t}))$.
Using  \citep[Proposition 7.44 p172]{BS}, both $h^{\varepsilon}_{G}$ and $p^{\varepsilon}_{G}$ are  $\mathcal{B}_{c}(\O^{t})$-measurable. For some $\o^{t} \in \O^{t}$, $y \geq 0$ fixed, if $h^{\varepsilon}(\o^{t},y) \notin  H_y^{t+1}(\o^{t},p^{\varepsilon}(\cdot,\o^{t},y))$,  using \eqref{overlinevt}, we have $\overline{\i}_{t}(\o^{t},h^{\varepsilon}(\o^{t},y),y,p^{\varepsilon}(\cdot,\o^{t},y))=-\infty< \min\left(\ \frac{1}{\varepsilon},I_{t}(\o^{t},y)-\varepsilon\right)$ (indeed from  \eqref{amiens2} at $t$, $I_{t} \geq 0$). This contradicts \eqref{Vt2t}   and  therefore $h^{\varepsilon}(\o^{t},y) \in H^{t+1}_{y}(\o^{t},p^{\varepsilon}(\cdot,\o^{t},y))$ and also $h^{\varepsilon}_{G} (\o^{t}) \in H^{t+1}_{G(\o^{t})}(\o^{t},p_{G}^{\varepsilon}(\cdot,\o^{t}))$
for $\o^{t} \in \left\{G \geq 0\right\}$. We set $P^{\varepsilon}_{G}:=P \otimes p^{\varepsilon}_{G} \in \mathcal{Q}^{t+1}$ (see \eqref{Qstar}) and get that
\begin{small}
\begin{align*}
P^{\varepsilon}_{G}(G(\cdot)+ h^{\varepsilon}_{G} (\cdot) \Delta S_{t+1}(\cdot) \geq 0) &= \int_{\{G \geq 0\}} \int_{\Omega_{t+1}} p^{\varepsilon}_{G}(G(\o^{t})+ h^{\varepsilon}_{G} (\o^{t}) \Delta S_{t+1}(\o^{t},\o_{t+1}) \geq 0,\o^{t}) P(d\o^{t})=1,
\end{align*}
\end{small}
\noindent since $\left\{G \geq 0\right\}$ is a $\mathcal{Q}^{t}$-full measure set, $h^{\varepsilon}_{G} \in \phi_{t+1}(G,P^{\varepsilon}_{G})$ follows.
Using \eqref{itusa} and \eqref{overlinevt},
\begin{scriptsize}
\begin{align*}
\int_{\O^{t}} \overline{\i}_{t}(\o^{t},h^{\varepsilon}_{G}(\o^{t}),p^{\varepsilon}_{G}(\o^{t}), G(\o^{t})) P(d\o^{t})  &=\int_{\Omega^{t+1}} I_{t+1}(\o^{t+1},G(\o^{t})+1+h^{\varepsilon}_{G}(\o^{t}) \Delta S_{t+1}(\o^{t+1})) P^{\varepsilon}_{G}(d\o^{t+1})
 \leq A, \end{align*}
\end{scriptsize}
where
$A:=\sup_{P \in \mathcal{Q}^{t+1}} \sup_{\xi \in \phi_{t+1}(G+1,P)} \int_{\Omega^{t+1}} I_{t+1}\left(\o^{t+1},G(\o^{t})+1+\xi(\o^{t}) \Delta S_{t+1}(\o^{t+1})\right) P(d\o^{t+1})$ and $A<\infty$ using \eqref{dimancheR} at $t+1$  ($\phi_{t+1}(G,P) \subset \phi_{t+1}(G+1,P)$).
Combining  with \eqref{Vt2t} we find that
\begin{align}
\nonumber \frac{1}{\varepsilon} \int_{\left\{I_{t}(\cdot, G(\cdot))=\infty\right\}} P(d\o^{t}) &+ \int_{\left\{I_{t}(\cdot, G(\cdot))<\infty\right\}} \left(I_{t}(\o^{t},G(\o^{t}))-\varepsilon\right) P(d\o^{t})\\
\label{star}
&  \leq  \int_{\O^{t}} \overline{\i}_{t}(\o^{t},h^{\varepsilon}_{G}(\o^{t}),G(\o^{t}),p^{\varepsilon}_{G}(\cdot,\o^{t})) P(d\o^{t})
 \leq A <\infty.
 \end{align}
As this is true for all $\varepsilon>0$,  $P(\left\{I_{t}(\cdot, G(\cdot))=\infty\right\})=0$ follows.  Using again \eqref{star}, we get that  $\int_{\Omega^{t}} I_{t}(\o^{t},\overline{G}(\o^{t-1}) + \xi(\o^{t-1}) \Delta S_{t}(\o^{t}))P(d\o^{t}) \leq A$ and as this is true for all $P \in \mathcal{Q}^{t}$ and $\xi \in \phi_{t}(\overline{G},P)$,
\eqref{dimancheR} is true for $t$.\\
We are left with the proof of  \eqref{noisette} for $U_t$.
Let $X=x+\sum_{s=1}^{t-1} \phi_s \Delta S_{s+1} $, with $x \geq 0$ and  $\left(\phi_s \right)_{1 \leq s \leq t-1}$ some  universally-predictable
random variables, be fixed such that $X\geq 0$ $\mathcal{Q}^{t}$-q.s. Let
$
{\Omega}^{t}_{X}:=\widetilde{\Omega}^{t} \cap \{X(\cdot) \geq 0\}.$ Then
${\Omega}^{t}_{X} \in \mathcal{B}_{c}(\O^{t})$  is a $\mathcal{Q}^{t}$-full measure set.
We introduce the following random set  $ {\psi}_{X}: \Omega^t \twoheadrightarrow \mathbb{R}^{d}$
\begin{scriptsize}
$$
 {\psi}_{X}(\o^{t}):= \left\{h \in \Dc_{X(\o^{t})}^{t+1}(\o^t),\; U_{t}(\o^{t}, X(\o^{t}))= \inf_{P \in \mathcal{Q}_{t+1}(\o^{t})} \int_{\Omega^{t+1}} U_{t+1}\left(\o^{t},\o_{t+1}, X(\o^{t})+h \Delta S_{t+1}(\o^{t},\o_{t+1})\right) P(d\o_{t+1}) \right\},
$$
\end{scriptsize}
for $\o^{t} \in {\Omega}^{t}_{X}$ and  ${\psi}_{X}(\o^{t})= \emptyset$ otherwise  ($\mathcal{D}^{t+1}_{X(\o^{t})}(\o^{t})$ is defined in  \eqref{domaineproj}).
To prove \eqref{noisette}, it is enough to find some $\mathcal{B}_{c}(\O^{t})$-measurable selector for ${\psi}_{X}$ and to show that $\Omega^{t}_{X} \subset \{\psi_{X} \neq \emptyset\}$.  The last point follows   from Proposition \ref{dyn1} and Theorem \ref{main1} (see \eqref{vopti}, \eqref{VQ}, \eqref{vanek}, \eqref{vanek2} and recall that $\Omega^{t}_{X} \subset \widetilde{\Omega}^{t}$).
Let  $u_{X}: \Omega^{t} \times \mathbb{R}^{d} \to \mathbb{R} \cup \{\pm\ \infty\}$ be defined by  (recall \eqref{widetildeu}) $u_{X}(\o^{t},h)=1_{{\Omega}_{X}^{t}}(\o^{t}) \widetilde{u}_{t}(\o^{t},h,X(\o^{t})).$
Using \citep[Proposition 14.39 p666,  Corollary 14.34 p664]{rw} we first prove that $-u_{X}$ is a  $\mathcal{B}_{c}(\O^t)$-normal integrand (see  \citep[Definition 14.27 p661]{rw}) and that  $u_{X}$ is $\mathcal{B}_{c}(\O^t)\otimes \mathcal{B}(\mathbb{R}^{d})$-measurable. Indeed  we show that for all $h \in \mathbb{R}^{d}$,  $u_{X}(\cdot,h)$ is $\mathcal{B}_{c}(\O^{t})$-measurable and for all $\o^{t} \in {\Omega}^{t}$, $u_{X}(\o^{t},\cdot)$ is usc and concave.
The first point follows from the fact that $\widetilde{u}_{t}$ is lsa, $X$ is $\mathcal{B}_{c}(\O^{t})$-measurable, ${\Omega}^{t}_{X} \in \mathcal{B}_{c}(\O^{t})$ and \citep[Proposition 7.44 p172]{BS}.  Now we fix $\o^{t} \in \O^{t}$. If $\o^{t} \notin {\Omega}_{X}^{t}$, it is clear that $u_{X}(\o^{t},\cdot)$ is usc and concave. If $\o^{t} \in  \Omega^{t}_{X} \subset \widetilde{\Omega}^{t}$, we know from Proposition  \ref{dyn1} that Lemma \ref{cont} applies and that    $\phi_{\o^{t}}(\cdot,\cdot)$ is usc and concave where  $
\phi_{\o^{t}}(x,h)=
\inf_{P \in \mathcal{Q}_{t+1}(\o^{t})}\int_{\O_{t+1} }U_{t+1}(\o^t, \o_{t+1}, x+ h \Delta S_{t+1}(\o^t, \o_{t+1}))P(d\o_{t+1})$ if $x \geq 0$ and $h\in \mathcal{H}_{x}^{t+1}(\o^{t})$ and  $-\infty$ otherwise.  In particular for $\o^{t} \in \O^{t}_{X}$ and  $x=X(\o^{t})$ we get that $ \phi_{\o^{t}}(X(\o^t),\cdot)=u_{X}(\o^{t},\cdot)$ is usc and concave.
Now, from the definitions of ${\psi}_{X}$ and $u_{X}$ for $\o^{t} \in \O^{t}_{X}$, we have that
$$ {\psi}_{X}(\o^{t})= \left\{h \in \mathcal{D}^{t+1}_{X(\o^{t})}(\o^{t}),\; U_{t}(\o^{t}, X(\o^{t}))= u_{X}(\o^{t},h)  \right\}.$$
Lemma \ref{Dhmes} implies that  $\mbox{Graph}\left(\mathcal{D}^{t+1}_{X}\right)  \in \mathcal{B}_{c}(\O^{t})\otimes \mathcal{B}(\mathbb{R}^{d})$.  Since $U_{t}$ is lsa,  $U_{t}$ is $\mathcal{B}_{c}(\O^{t} \times \mathbb{R})$-measurable and   \citep[Lemma 7.29 p174]{BS}  implies that $U_{t}(\cdot,x)$  is $\mathcal{B}_{c}(\O^{t})$-measurable for $x \in \mathbb{R}$ fixed. From \eqref{amiens} $U_{t}(\o^{t},\cdot)$ is usc and nondecreasing for any fixed  $\o^{t} \in \O^{t}$, so  \citep[Lemmata 7.12, 7.16]{BCR} implies that $U_{t}$ is $\mathcal{B}_{c}(\O^t)\otimes \mathcal{B}(\mathbb{R})$-measurable. As $X$ is  $\mathcal{B}_{c}(\O^{t})$-measurable, we obtain that $U_{t}(\cdot,X(\cdot))$ is $\mathcal{B}_{c}(\O^{t})$-measurable (see  \citep[Proposition 7.44 p172]{BS}). It follows that
 $\mbox{Graph}(\psi_{X}) \in \mathcal{B}_{c}(\O^{t})\otimes \mathcal{B}(\mathbb{R}^{d})$, we can apply  the Projection Theorem (see  \citep[Theorem 3.23 p75]{CV77}) and we get that $\left\{\psi_{X} \neq \emptyset \right\} \in \mathcal{B}_{c}(\O^{t})$. Using Auman Theorem (see  \citep[Corollary 1]{bv})  there exists some $\mathcal{B}_{c}(\O^{t})$-measurable $\widehat{h}_{t+1}^{X}: \left\{\psi_{X} \neq \emptyset \right\}\to \mathbb{R}^{d}$  such that for all $\o^{t} \in \left\{\psi_{X} \neq \emptyset \right\}$, $\widehat{h}_{t+1}^{X}(\o^{t}) \in \psi_{X}(\o^{t})$. This concludes the proof of \eqref{noisette} extending $\widehat{h}_{t+1}^{X}$ on all $\Omega^{t}$ ($\widehat{h}_{t+1}^{X}=0$ on $ \Omega^t \setminus \left\{\psi_{X} \neq \emptyset \right\}$).

\end{proof}\\

\begin{proof}\emph{of Theorem \ref{main}.} We proceed in three steps. First, we handle some integrability issues that are essential to the proof and where not required in \citep{Nutz}. In particular we show that it is possible to apply Fubini Theorem. Then, we build by induction a candidate for the optimal strategy and  finally we establish its optimality. The proof of the two last steps is very similar to the one of \citep{Nutz}. \\
\textbf{Integrability Issues}\\
First from Proposition \ref{propufini} and \eqref{eq:OP}, $u(x)\leq M_x<\infty$.
We fix some $x \geq 0$ and $\phi \in \Phi(x, \mathcal{Q}^{T})=\Phi(x,U,\mathcal{Q}^{T})$ (see again Proposition \ref{propufini}). From Proposition \ref{dyn2}, we can apply by backward induction Proposition \ref{dyn3} for $t=T-1,T-2, \dots, 0$. In particular, we get that \eqref{reims3} and \eqref{dimancheR}  hold true for all $ 0 \leq t \leq T$ and choosing $G=V_{t-1}^{x+1,\phi}$ and  $\xi=\phi_{t}$ (use Lemma \ref{AOAT} since $\phi \in \Phi(x, \mathcal{Q}^{T})$),  we get for all $P \in \mathcal{Q}^{t}$,
\begin{align}
\label{fubiniforphi0}
\int_{\Omega^{t}} U^{+}_{t}\left(\omega^{t},V_{t}^{x,\phi} (\o^t)\right) P(d\o^{t})<\infty.
\end{align}
So for all $P=P_{t-1} \otimes p \in \mathcal{Q}^{t}$ (see \eqref{Qstar}) \citep[Proposition 7.45 p175]{BS} implies that \begin{small}\begin{align}
\label{fubiniforphi}
\int_{\Omega^{t}} U_{t}\left(\omega^{t},V_{t}^{x,\phi} (\o^t)\right) P(d\o^{t})=\int_{\Omega^{t-1}} \int_{\Omega_{t}} U_{t}\left(\omega^{t-1},\o_{t},V_{t}^{x,\phi} (\o^{t-1},\o_{t})\right)p(d\o_{t}, \o^{t-1}) P_{t-1}(d\o^{t-1}).
\end{align}
\end{small}
\noindent \textbf{Construction of $\phi^{*}$}\\
We fix some $x \geq 0$ and  build by  induction our candidate $\phi^*$ for the optimal strategy which will verify that
\begin{footnotesize}
\begin{align}
\label{travi}
U_{t}\left(\omega^{t},V_{t}^{x,\phi^*} (\o^t)\right)= \inf_{P \in \mathcal{Q}_{t+1}(\o^{t})}\int_{\Omega_{t+1}} U_{t+1}\left(\omega^{t},\omega_{t+1},V_{t}^{x,\phi^*} (\o^t) + \phi^*_{t+1}(\o^t) \Delta S_{t+1}(\o^t,\o_{t+1})\right)P(d\o_{t+1}).
\end{align}
\end{footnotesize}
We start at $t=0$ and use \eqref{noisette} in Proposition \ref{dyn3} with $X=x \geq 0$. We set $\phi^*_1:=\widehat{h}^x_{1} \in \mathcal{D}^{1}_{x}$ and we obtain that  $P_{1}(x+ \phi^*_1 \Delta S_1(.)\geq 0)=1$ for all $P \in \mathcal{Q}^1$ and that \eqref{travi} holds true for $t=0$.
Assume that until some $t \geq 1$ we have found some  universally-predictable  random variables $\left(\phi^*_s \right)_{1 \leq s \leq t}$  and some sets $\left(\overline{\Omega}^s \right)_{1 \leq s \leq t-1}$ such that $\overline{\Omega}^{s} \in \mathcal{B}_{c}(\O^{s})$ is a $\mathcal{Q}^{s}$-full measure set,
 $\phi^{*}_{s+1}(\o^s) \in D^{s+1}(\o^{s})$ for all $\o^{s} \in \overline{\Omega}^s$,  $\{V^{x,\phi^*}_{s+1}(\cdot)\geq 0\}$ is a  ${Q}^{s+1}$-full measure set  and  \eqref{travi} holds true at $s$ for all $\o^{s} \in \overline{\Omega}^{s}$ where $s=0,\dots,t-1$.
 We  apply Proposition \ref{dyn3} with  $X= V_{t}^{x,\phi^*}$  and there exists $\mathcal{Q}^{t}$-full measure set $\overline{\Omega}^{t}:={\Omega}^{t}_{V_{t}^{x,\phi^*}} \in \mathcal{B}_{c}(\O^{t})$  and  some $\mathcal{B}_{c}(\O^{t})$-measurable random variable $\phi^*_{t+1}:=\widehat{h}^{V^{x,\phi^*}_t}_{t+1}$  such that  $\phi^{*}_{t+1}(\o^{t}) \in \mathcal{D}^{t+1}_{V^{x,\phi^*}_{t}(\o^{t})}(\o^{t})$ for all $\omega^{t} \in \overline{\Omega}^{t}$  and \eqref{travi} holds true at $t$.  Let  $P^{t+1}=P \otimes p \in \mathcal{Q}^{t+1}$ where $P \in \mathcal{Q}^{t}$ and $p \in \mathcal{S}K_{t+1}$ with  $p(\cdot,\o^{t}) \in \mathcal{Q}_{t+1}(\o^{t})$ for all $\o^{t} \in \overline{\O}^{t}$ (see \eqref{Qstar}).
From \citep[Proposition 7.45 p175]{BS} we get
$$P_{t+1}(V_{t+1}^{x,\phi^*} \geq 0)=  \int_{\Omega^{t}}p(V_{t}^{x,\phi^*}(\o^{t}) + \phi^*_{t+1}(\o^{t}) \Delta S_{t+1}(\o^{t},\cdot) \geq 0, \o^{t}) P(d\omega^{t})=1,$$
where we have used that  $\phi^{*}_{t+1}(\o^{t}) \in \mathcal{H}^{t+1}_{V^{x,\phi^*}_{t}(\o^{t})}(\o^{t})$ for all $\omega^{t} \in \overline{\Omega}^{t}$ and $P(\overline{\O}^{t})=1$
and we can continue the recursion.  Thus, we have found  that $\phi^{*} \in \Phi(x, \mathcal{Q}^{T})$ and  from Proposition \ref{propufini}, $\phi^{*} \in \Phi(x,U, \mathcal{Q}^{T})$.\\
 \textbf{Optimality of $\phi^*$}\\
We fix some $P=P_{T-1} \otimes p_{T} \in \mathcal{Q}^{T}$. Using \eqref{fubiniforphi}, $P_{T-1}(\overline{\Omega}^{T-1})=1$ and \eqref{travi} for $t=T-1$  we get that
\begin{scriptsize}
\begin{align*}E_{P}U(\cdot,V_T^{x,\phi^*}(\cdot))
&=\int_{\overline{\Omega}^{T-1}} \int_{\Omega_{T}} U_{T}\left(\omega^{T-1}, \omega_{T}, V_{T-1}^{x,\phi^*}(\omega^{T-1}) +\phi^{*}_{T}(\omega^{T-1})  \Delta S_T(\omega^{T-1}, \omega_{T})\right) p_{T}(d\omega_{T},\o^{T-1}) P_{T-1}(d\o^{T-1})\\
& \geq \int_{{\Omega}^{T-1}} U_{T-1}\left(\omega^{T-1}, V_{T-1}^{x,\phi^*}(\omega^{T-1})\right) P_{T-1}(d\o^{T-1}).
\end{align*}
\end{scriptsize}
We iterate the process by backward induction and obtain that (recall that $\Omega^{0}:=\{\o_{0}\}$)
$
U_0(x) \leq E_{P}U(\cdot,V_T^{x,\phi^*}(\cdot))$.
As the preceding equality holds true for all $P \in \mathcal{Q}^{T}$ and as $\phi^* \in \Phi(x,U, \mathcal{Q}^{T})$, we get that  $U_0(x)\leq u(x)$ (see \eqref{eq:OP}). So $\phi^* $ will be optimal if $U_0(x)\geq u(x).$
We fix  some $\phi \in \Phi(x,U, \mathcal{Q}^{T})$ and show that
\begin{align}
\label{last}
 \inf_{P \in \mathcal{Q}^{t+1}} E_{P}  U_{t+1}(\cdot,V_{t+1}^{x,\phi}(\cdot))  \leq  \inf_{Q \in \mathcal{Q}^{t}}E_{Q} U_{t}(\cdot,V_{t}^{x,\phi}(\cdot)), \; t\in\{0,\ldots,T-1\}.
 \end{align}
Then $\inf_{P \in \mathcal{Q}^{T}} E_{P}  U_{T}(\cdot,V_{T}^{x,\phi}(\cdot)) \leq \inf_{Q \in \mathcal{Q}^{1}} E_{Q}  U_{1}(\cdot,V_{1}^{x,\phi}(\cdot))  \leq U_{0}(x)$ is obtained recursively (recall \eqref{Clot}).
As this is true for all $\phi \in \Phi(x,U, \mathcal{Q}^{T})$,
$ u(x) \leq U_{0}(x)$
and the proof is complete.\\
We fix some $t\in\{0,\ldots,T-1\}$ and prove \eqref{last}. As $U_{t+1}$ is lsa (see \eqref{reims})  and Assumption \ref{Sass} holds true,  \citep[Lemma 7.30 (3) p177, Proposition 7.48 p180]{BS} imply that $f$ is lsa where
$$f(\o^{t},y,h,P) := \int_{\O_{t+1}} U_{t+1}(\o^{t},\o_{t+1},y+ h \Delta S_{t+1}(\o^{t},\o_{t+1})) P(d\o_{t+1}).$$
Let $f^*(\o^{t},y,h) = \inf_{P \in \mathcal{Q}_{t+1}(\o^{t})} f(\o^{t},y,h,P)$ and fix some $\varepsilon>0$. Then since
$\{(\o^{t},y,h,P),  P \in  \mathcal{Q}_{t+1}(\o^{t})\}$ is an analytic set (recall Assumption \ref{Qanalytic}), \citep[Proposition 7.50 p184]{BS} implies that there exists some universally-measurable  $\widetilde{p}^{\varepsilon}_{t+1}: (\o^{t},y,h)  \to \mathfrak{P}(\O_{t+1})$ such that $\widetilde{p}^{\varepsilon}_{t+1}(\cdot,\o^{t},y,h) \in  \mathcal{Q}_{t+1}(\o^{t})$ for all $(\o^{t},y,h) \in \O^{t} \times \mathbb{R} \times \mathbb{R}^{d}$ and
\begin{align}
\label{equt1}
 f(\o^{t},y,h,\widetilde{p}^{\varepsilon}_{t+1}(\cdot,\o^{t},y,h))
&\leq \begin{cases} f^*(\o^{t},y,h) + \varepsilon,
 \mbox{ if $ f^*(\o^{t},y,h)>-\infty$}\\
-\frac{1}{\varepsilon}, \; \mbox{otherwise}.
\end{cases}
\end{align}
Let $p^{\varepsilon}_{t+1}(\cdot,\o^{t})= \widetilde{p}^{\varepsilon}_{t+1}\left(\cdot,\o^{t},V_{t}^{x,\phi}(\o^{t}), \phi_{t+1}(\o^{t})\right)$: \citep[Proposition 7.44 p172]{BS} implies that $p^{\varepsilon}_{t+1}$ is  $\mathcal{B}_{c}(\O^{t})$-measurable.
For all $\o^{t} \in \widetilde{\Omega}^{t} \cap \{\;V_{t}^{x,\phi}(\cdot) \geq 0\}$, $ f^*(\o^{t},V_{t}^{x,\phi}(\o^{t}),\phi_{t+1}(\o^{t})) \leq
\sup_{h \in \mathcal{H}^{t+1}_{V_{t}^{x,\phi}(\o^{t})}(\o^{t}) }  f^*(\o^{t},V_{t}^{x,\phi}(\o^{t}),h) =  U_{t}(\o^{t}, V_{t}^{x,\phi}(\o^{t}))$
(use  Lemma \ref{AOAT} since $\phi \in \Phi(x, \mathcal{Q}^{T})$ and recall \eqref{Clot}). Choosing $y= V_{t}^{x,\phi}(\o^{t})$, $h=\phi_{t+1}(\o^{t})$ in \eqref{equt1}, we find that   for all $\o^{t} \in \widetilde{\Omega}^{t} \cap \{V_{t}^{x,\phi}(\cdot) \geq 0\}$
\begin{align}
\label{equt2}
\int_{\Omega_{t+1}} U_{t+1}(\o^{t},\o_{t+1},V_{t+1}^{x,\phi}(\o^{t},\o_{t+1}))p^{\varepsilon}_{t+1}(d\o_{t+1},\o^{t}) -\varepsilon \leq \max \left(U_{t}(\o^{t},V_{t}^{x,\phi}(\o^{t})), -\frac{1}{\varepsilon} -\varepsilon \right).
\end{align}
Fix some $Q \in \mathcal{Q}^{t}$  and set $P^{\varepsilon}:=Q \otimes p^{\varepsilon}_{t+1} \in \mathcal{Q}^{t+1}$ (see \eqref{Qstar}).  Using \eqref{fubiniforphi} and since $\widetilde{\Omega}^{t} \cap \{V_{t}^{x,\phi}(\cdot) \geq 0\}$ is a $\mathcal{Q}^{t}$ full measure set (recall again that $\phi \in \Phi(x,\mathcal{Q}^{T})$ and Lemma \ref{AOAT}) , we get
$$
\inf_{P \in \mathcal{Q}^{t+1}} E_{P}  U_{t+1}(\cdot,V_{t+1}^{x,\phi}(\cdot)) -\varepsilon  \leq  E_{P^{\varepsilon}} U_{t+1}(\cdot,V_{t+1}^{x,\phi}(\cdot)) -\varepsilon
 \leq E_{Q} \max \left(U_{t}(\cdot,V_{t}^{x,\phi}(\cdot)), -\frac{1}{\varepsilon}-\varepsilon\right).
$$
Since for all $0< \varepsilon<1$,  $\max\left(U_{t}(\cdot,V_{t}^{x,\phi}(\cdot)), -\frac{1}{\varepsilon}-\varepsilon\right) \leq -1+ U^{+}_{t}(\cdot,V_{t}^{x,\phi}(\cdot))$, recalling \eqref{fubiniforphi0}, letting $\varepsilon$ go to zero and applying Fatou's Lemma, we obtain that
$ \inf_{P \in \mathcal{Q}^{t+1}} E_{P}  U_{t+1}(\cdot,V_{t+1}^{x,\phi}(\cdot))  \leq    E_{Q} U_{t}(\cdot,V_{t}^{x,\phi}(\cdot)).$
As this holds true for all $Q \in \mathcal{Q}^{t}$,  \eqref{last} is proved.
\end{proof}\\

\begin{proof}\emph{of Theorem \ref{main2}.}
Since the $sNA(\mathcal{Q}^{T})$ condition holds true, the $NA(\mathcal{Q}^{T})$ condition is also verified and  to apply Theorem \ref{main} it remains to  prove that Assumption \ref{Uplus} is satisfied. We fix some $P \in \mathcal{Q}^{T}$ $x \geq 0$  and some $\phi \in \phi(x,P)$.
Since the $NA(P)$ condition holds true, using similar arguments  as in the proof of  \citep[Theorem 4.17]{BCR} we find that for ${P}_{t}$-almost all $\o^{t} \in \O^{t}$, $
|V_{t}^{x,\phi}(\o^{t})| \leq   \prod_{s=1}^{t}\left(x+ \frac{|\Delta S_{s}(\o^{s})|}{\alpha^{P}_{s-1}(\o^{s-1})}\right)$.
Note  that $V^{x,\phi} $ is universally-adapted and that
$\sup_{P \in \mathcal{Q}^{t}}E_{P} |V_{t}^{x,\phi}(\cdot)|^{r} <\infty$ for all $r>0$ (recall that $\Delta S_{s}, \;\frac{1}{\alpha_{s}^{P}} \in \mathcal{W}_{s}$ for all $s \geq 1$).  The monotonicity of $U^{+}$ and Proposition \ref{ae} (with $\l= 2\prod_{s=1}^{T}\left(1+ \frac{|\Delta S_{s}(\o^{s})|}{\alpha^{P}_{s-1}(\o^{s-1})}\right) \geq 1$) implies that for ${P}_{t}$-almost all $\o^{t} \in \O^{t}$
\begin{align}
\label{eqJtQ}
U^{+}(\o^{T}, V_{T}^{1,\phi}(\o^{T})) \leq 4 \left( \prod_{s=1}^{T}  \left(1+\frac{|\Delta S_{s}(\o^{s})|}{\alpha^{P}_{s-1}(\o^{s-1})}\right)\right)\left( U^{+}(\o^{T},1) +C_{T}(\o^{T})\right).
\end{align}
We set $N:=4  \sup_{P \in \mathcal{Q}^{T}} E_{P} \left(\left( \prod_{s=1}^{T}  \left(1+\frac{|\Delta S_{s}(\o^{s})|}{\alpha^{P}_{s-1}(\o^{s-1})}\right)\right)\left( U^{+}(\o^{T},1) +C_{T}(\o^{T})\right)\right)$.
Since $U^{+}(\cdot,1)$, $U^{-}(\cdot,\frac{1}{4})$ $\in$ $\mathcal{W}_{T}$ and $\Delta S_{s}, \; \frac{1}{\alpha_{s}^{P}} \in \mathcal{W}_{s}$ for all $s \geq 1$,
we obtain that $N<\infty$ (recall the definition of $C_{T}$ in Proposition \ref{ae}). Using \eqref{eqJtQ}  we find that $E_{P} U^{+}(\cdot, V_{T}^{1,\phi}(\cdot))  \leq N<\infty$ and as this is true  for all $P \in \mathcal{Q}^{T}$ and $\phi \in \Phi(1,P)$, Assumption \ref{Uplus} holds true.
\end{proof}
\section{Appendix}
\label{seannexe}

\subsection{\textbf{Auxiliary results}}
\label{proofofres}
The two first Lemmata were used in the proof of Theorem \ref{main1} and Lemma \ref{Dhmes}.  The second one is a well-know result on concave functions which  proof is given since we did not find some reference.
\begin{lemma}
\label{rast2}
Assume that Assumption \ref{Yb} holds true. For all  $x > 0$, we have $\mbox{Aff}(\mathcal{H}_{x})=\mathbb{R}^{d}$, $\mbox{Ri}(\mathcal{H}_{x})$ is an open set in $\mathbb{R}^{d}$ and $\mathbb{Q}^{d}$ is dense in $\mbox{Ri}(\mathcal{H}_{x})$ \footnote{For a Polish space $X$, we say that a set $D \subset X$ is dense in $B \subset X$ if for all $\varepsilon>0$, $b \in B$, there exists $d \in D\cap B$ such that $d(b,d) < \varepsilon$ where $d$ is a metric on $X$ consistent with its topology.} . Moreover   $
  \mbox{Ri}(\mathcal{H}_{x}) \subset   \bigcup_{\substack{r \in \mathbb{Q},\;r>0}} \mathcal{H}^{r}_{x} \subset \mathcal{H}_{x}$
 and therefore
$ \overline{ \bigcup_{\substack{r \in \mathbb{Q},\;r>0}} \mathcal{H}^{r}_{x}}= \mathcal{H}_{x}, \mbox{where the closure is taken in $\mathbb{R}^{d}$.}$  If furthermore, we assume that there exists some $0 \leq c <\infty$ such that $Y_{i}(\o) \leq c$ for all $i=1,\cdots,d$, $\o \in \overline{\O}$ (recalling Assumption \ref{Yb}, $|Y|$ is bounded)  then $
 \mbox{Ri}(\mathcal{H}_{x}) =   \bigcup_{\substack{r \in \mathbb{Q},\;r>0}} \mathcal{H}^{r}_{x}$.
 \end{lemma}
\begin{proof}
Fix some $x > 0$. Let $\varepsilon>0$ be such that $x-\varepsilon>0$ and $R:=\{h \in \mathbb{R}^{d}, \;  0 \leq h_{i} \leq \frac{x-\varepsilon}{db}\}$.  Using Assumption \ref{Yb},  if $h \in R$ for all $\o \in \overline{\Omega}$,
$x + h Y(\o) \geq x -b \sum_{i=1}^{d} h_{i} \geq \varepsilon$  and $h \in \mathcal{H}^{\varepsilon}_{x} \subset \mathcal{H}_{x}$. Thus  $R \subset \mathcal{H}_{x}$ and $\mbox{Aff}(\mathcal{H}_{x})=\mathbb{R}^{d}$  follows (recall that $0 \in \mathcal{H}_{x}$). Therefore $\mbox{Ri}(\mathcal{H}_{x})$ is the interior of $\mathcal{H}_x$ in $\mathbb{R}^{d}$ and thus an open set in $\mathbb{R}^{d}$ and the fact that $\mathbb{Q}^{d}$ is dense in $\mbox{Ri}(\mathcal{H}_{x})$ follows immediately.
 Fix now some $h \in \mbox{Ri}(\mathcal{H}_{x})$. As $0 \in \mathcal{H}_{x}$, there exists some $\varepsilon>0$ such that $(1+\varepsilon)h \in  \mathcal{H}_{x}$, see  \citep[Theorem 6.4 p47]{cvx} which implies  that $x+ h Y(\cdot) \geq \frac{\varepsilon}{1+\varepsilon} x>0$ $\mathcal{Q}$-q.s., hence  $h \in  \mathcal{H}^{r}_{x}$ for $r \in \mathbb{Q}$ such that $0<r\leq  \frac{\varepsilon}{1+\varepsilon} x$ and  $
  \mbox{Ri}(\mathcal{H}_{x}) \subset   \bigcup_{\substack{r \in \mathbb{Q},\;r>0}} \mathcal{H}^{r}_{x} \subset \mathcal{H}_{x}$ is proved and also
$ \overline{ \bigcup_{\substack{r \in \mathbb{Q},\;r>0}} \mathcal{H}^{r}_{x}}= \mathcal{H}_{x}$ since $ \overline{\mbox{Ri}(\mathcal{H}_{x})}=\mathcal{H}_{x}$. Assume now that $|Y|$ is bounded by some constant $K>0$.  Let $h \in   \bigcup_{\substack{r \in \mathbb{Q},\;r>0}} \mathcal{H}^{r}_{x}$ and $r \in \mathbb{Q}$, $r>0$ be such that $h \in \mathcal{H}^{r}_{x}$, we set $\varepsilon:= \frac{r}{2K}$. Then for any $g \in B(0,\varepsilon)$, we have for  $\mathcal{Q}$-almost all $\o \in \overline{\O}$ that $x+(h+g)Y(\o) \geq r + gY(\o) \geq r- |g||Y(\o)| \geq \frac{r}{2}$, hence $h+g \in \mathcal{H}_{x}$, $B(h,\varepsilon) \subset \mathcal{H}_{x}$ and $h$ belongs to the interior of $\mathcal{H}_{x}$ (and also to $\mbox{Ri}(\mathcal{H}_{x})$).
\end{proof}\\

\begin{lemma}
\label{supconvex}
Let $f: \mathbb{R}^{d} \to \mathbb{R} \cup \{\pm \infty\}$ be a concave function such that $\mbox{Ri}(\mbox{Dom} \, f ) \neq \emptyset$.
Then $\sup_{h \in \textnormal{Dom} \, f} f(h)=\sup_{h \in \textnormal{Ri}(\textnormal{Dom} \, f)} f(h).$
\end{lemma}
\begin{proof}
Let $C:=\sup_{h \in \textnormal{Ri}(\textnormal{Dom} \,f )} f(h)$ and $h_{1} \in \mbox{Dom} \, f \backslash \mbox{Ri}(\mbox{Dom} \,f) $ be fixed. We have to  prove that $f(h_{1}) \leq C$.
If $C=\infty$ there is nothing to show. So assume that $C<+\infty$.  Let $h_{0}\in \mbox{Ri}(\mbox{Dom} \, f )$ and introduce   $\phi : t\in \mathbb{R} \rightarrow f(th_{1}+(1-t)h_{0})$ if $t \in [0,1]$ and $-\infty$ otherwise.  From  \citep[Theorem 6.1 p45]{cvx},  $th_{1}+(1-t)h_{0} \in \mbox{Ri}(\mbox{Dom} \, f )$ if $t \in [0,1)$ and thus $[0,1) \subset \{t \in [0,1], \phi(t)\leq C\}$. Clearly, $\phi$ is concave on $\mathbb{R}$. Since $\mbox{Dom} \, f $ is convex,  $\mbox{Dom} \, \phi =[0,1]$.  So, using  \citep[Proposition A.4 p400]{fs}, $\phi$ is lsc on $[0,1]$ and $\{t \in [0,1],\;\phi(t) \leq C\}$ is a closed set in $\mathbb{R}$. It follows that  $1 \in \{t \in [0,1], \phi(t)\leq C\}$, $i.e$ $f(h_{1})\leq C$ and the proof is complete.
\end{proof}\\

The following lemma was used several times.
\begin{lemma}
\label{AOAT}
Assume that  the $NA(\mathcal{Q}^{T})$ condition holds true. Let $\phi \in \Phi$ such that $V_T^{x,\phi} \geq 0$ $\mathcal{Q}^{T}$-q.s. (i.e. $\phi \in \Phi(x, \mathcal{Q}^{T})$), then
$V_t^{x,\phi} \geq 0$ $\mathcal{Q}^{t}$-q.s. for all $t\in \{0,\ldots,T\}$.
\end{lemma}
\begin{proof}
Let $\phi \in \Phi$ be such that $V_T^{x,\phi} \geq 0$ $\mathcal{Q}^{T}$-q.s. and assume that $V_t^{x,\phi} \geq 0$ $\mathcal{Q}^{t}$-q.s. for all $t$ does not hold true.  Then $n:=\sup\{t, \; \exists P_{t} \in \mathcal{Q}^{t},\; P_t(V_{t}^{{x,\phi}}<0)>0\}<T$ and  there exists some $\widehat{P}_{n} \in \mathcal{Q}^{n}$ such that  $\widehat{P}_n(A)>0$ where $A=\{V_{n}^{x,\phi}<0\} \in \mathcal{B}_{c}(\O^{n})$ and for all $s \geq n+1$, $P \in \mathcal{Q}^{s}$, $P(V_{s}^{x,\phi} \geq 0)=1$.
Let $\Psi_s(\o^{s-1})=0$ if  $1 \leq s\leq n$ and  $\Psi_s(\o^{s-1})=1_{A}(\o^{n}) \phi_s(\o^{s-1})$ if $s\geq n+1$. Then $\Psi \in \Phi$ and
$
{V}^{0,\Psi}_T =  \sum_{k=n+1}^T \Psi_s \Delta  S_s
 =  1_A\left({V}^{x,\phi}_T-{V}^{x,\phi}_n\right)$. Thus $V_{T}^{0,\Psi}  \geq 0$ $\mathcal{Q}^{T}$-q.s. and $V_{T}^{0,\Psi}  > 0$ on $A$. Let $\widehat{P}_{T}:=\widehat{P}_{n} \otimes p_{n+1} \cdots \otimes p_{T} \in \mathcal{Q}^{T}$ where for $s=n+1,\cdot,T$,  $p_{s}(\cdot,\cdot)$ is a given universally-measurable selector of $\mathcal{Q}^{s}$ (see \eqref{Qstar}). It is clear that  $\widehat{P}_T(A)=\widehat{P}_n(A)>0$, hence we get an arbitrage opportunity.
\end{proof}\\

\subsection{\textbf{Measure theoretical issues}}
\label{apap}
In this section, we first provide some counterexamples to  \citep[Lemma 4.12]{BN} and   propose  an alternative to this lemma. Our counterexample \ref{uscusc} is based on a result from \citep{GO1} originally due  \citep{Siep}. An other  counterexample can be found  \citep[Proposition 14.28 p661]{rw}. \begin{example}
\label{uscusc}
 We denote by $\mathcal{L}(\mathbb{R}^{2})$ the Lebesgue sigma-algebra on $\mathbb{R}^{2}$. Recall that $\mathcal{B}(\mathbb{R}^{2}) \subset \mathcal{L}(\mathbb{R}^{2}$). Let $A \notin \mathcal{L}(\mathbb{R}^{2})$  be such that every line has at most two common points with $A$ (see  \citep[Example 22 p142]{GO1} for the proof of the existence of $A$) and define $F: \mathbb{R}^{2} \to \mathbb{R}$ by
$F(x,y):=1_{A}(x,y).$
We fix some $x \in \mathbb{R}$ and let $A^{1}_{x}:=\{y \in \mathbb{R},\; (x,y) \in A\}$.  By assumption, $A^{1}_{x}$ contains at most two points: thus it is a closed subset of $\mathbb{R}$. It follows that  $\{y \in \mathbb{R},\; F(x,y)\geq c\}$  is a closed subset of $\mathbb{R}$ for all $c \in \mathbb{R}$ and $F(x,\cdot)$ is usc.
Similarly the function $F(\cdot,y)$ is usc and 
 thus $\mathcal{B}(\mathbb{R})$-measurable for all $y \in \mathbb{R}$ fixed. But since $A \notin \mathcal{L}(\mathbb{R}^{2})$, $F$ is not $\mathcal{L}(\mathbb{R}^{2})$-measurable and therefore not $\mathcal{B}(\mathbb{R}) \otimes \mathcal{B}(\mathbb{R}) $-measurable.
\end{example}
We propose now the following correction to  \citep[Lemma 4.12]{BN}.  Note that Lemma \ref{412} can be applied  in the proof of  \citep[Lemma 3.7]{Nutz}  since the considered function  is concave (as well as in the proof of  \citep[Lemma 4.10]{BN} where the considered function  is  convex).
\begin{lemma}
\label{412}
Let $(A,\mathcal{A})$ be a measurable space and let $\theta: \mathbb{R}^{d} \times A \to \mathbb{R} \cup \{\pm \infty\}$ be a function such that $\o \to \theta(y,\o)$ is $\mathcal{A}$-measurable for all $y \in \mathbb{R}^{d}$ and $y \to \theta(y,\o)$ is lsc and convex for all $\o \in A$. Then $\theta$ is $\mathcal{B}(\mathbb{R}^{d})\otimes \mathcal{A}$-measurable.
\end{lemma}
\begin{proof}
It is a direct application of \citep[Proposition 14.39 p666,  Corollary 14.34 p664]{rw}.
\end{proof}\\

 We  finish with three lemmata related to measurability issues used throughout the paper.
 \begin{lemma}
  \label{partial}
Let $X,Y$ be two Polish spaces and $F: X \times Y \to \mathbb{R} \cup \{\pm \infty\}$ be usa (resp. lsa). Then, for  $x \in X$ fixed, the function $F_{x}: y \in Y  \to F(x,y)\in  \mathbb{R} \cup \{\pm \infty\}$ is usa (resp. lsa).
  \end{lemma}
  \begin{proof}
Assume that $F$ is usa and fix some $c \in \mathbb{R}$, then
$ C:=F^{-1}((c,\infty))  \in \mathcal{A}(X\times Y).$
Fix now some  $x \in X$. Since $I_{x}: y \to (x,y) $ is $\mathcal{B}(Y)$-measurable, applying  \citep[Proposition 7.40 p165]{BS}, we get that
 $\{y \in Y,\; F_{x}(y) >c\} =\{ y \in Y,\; (x,y) \in C\}= I_{x}^{-1}(C) \in \mathcal{A}(Y).$
  \end{proof}\\

\begin{lemma}
\label{LemmaAA2}
Assume that Assumptions \ref{Qanalytic} and \ref{Sass} hold true.
Let $0\leq t \leq T-1$,  $B \in \mathcal{B}(\mathbb{R})$. Then
\begin{align*}
F_{B}:(\o^{t},P,h,x)  & \to  P\left(x+h\Delta S_{t+1}(\o^t,\cdot)\in B\right) \mbox{ is $\mathcal{B}(\Omega^{t})\otimes \mathcal{B}(\mathfrak{P}(\O_{t+1})) \otimes\mathcal{B}(\mathbb{R}^{d})\otimes\mathcal{B}(\mathbb{R})$-measurable}\\
H_{B}: (\omega^{t},h,x) & \to  \inf_{P \in \mathcal{Q}_{t+1}(\o^{t})} P(x+h\Delta S_{t+1}(\o^t,\cdot) \in B) \in \mathcal{L}SA(\Omega^{t}\times\mathbb{R}^{d}\times \mathbb{R})\\
K_{B}: (\omega^{t},h)   & \to  \sup_{P \in \mathcal{Q}_{t+1}(\o^{t})} P(x+h\Delta S_{t+1}(\o^t,\cdot) \in B) \in \mathcal{U}SA(\Omega^{t}\times\mathbb{R}^{d}).
 \end{align*}
\end{lemma}
\begin{proof}
The first assertion follows  from  \citep[Proposition 7.29 p144]{BS} applied to $f(\o_{t+1},\o^{t},P,h,x)$ $=$ $1_{x+h\Delta S_{t+1}(\o^t,\cdot)\in B}(\o_{t+1})$ (recall  Assumption \ref{Sass}) and $q(d\o_{t+1}|\o^{t},P,h,x)=P(d\o_{t+1})$. The second one is obtained applying \citep[Proposition 7.47 p179]{BS} to $F_{B}$  (recall Assumption \ref{Qanalytic}). The last assertion is using  $ \sup_{P \in \mathcal{Q}_{t+1}(\o^{t})} P(x+h\Delta S_{t+1}(\o^t,\cdot) \in B)=1 - \inf_{P \in \mathcal{Q}_{t+1}(\o^{t})} P(x+h\Delta S_{t+1}(\o^t,\cdot) \in B^{c})$  and  Lemma \ref{partial}.
\end{proof}\\

\begin{lemma}
\label{lambda}
Let $X$ be a Polish space and $\Lambda$ be an $\mathbb{R}^{d}$-valued  random variable.
\begin{itemize}
\item[i)] Assume that $\mbox{Graph}(\Lambda) \in \mathcal{B}_{c}(X)\otimes \mathcal{B}(\mathbb{R}^{d})$. Then
$\mbox{Graph}(\overline{\Lambda}) \in \mathcal{B}_{c}(X)\otimes \mathcal{B}(\mathbb{R}^{d})$ where $\overline{\Lambda}$ is defined by  $\overline{\Lambda}(x)=\overline{\Lambda(x)}$ for all $x \in X$ (where the closure is taken in $\mathbb{R}^{d}$).
\item[ii)] Assume now that $\Lambda$ is open valued and $\mbox{Graph}(\Lambda) \in  \mathcal{C}A(X \times \mathbb{R}^{d})$. Then
$\mbox{Graph}(\Lambda) \in \mathcal{B}_{c}(X)\otimes \mathcal{B}(\mathbb{R}^{d})$.
\end{itemize}
\end{lemma}
\begin{proof} From  \citep[Theorem 14.8 p648]{rw}, $\Lambda$ is  $\mathcal{B}_{c}(X)$-measurable (see \citep[Definition 14.1 p643]{rw})  and using  \citep[Theorem 18.6 p596]{Hitch} we get that $\mbox{Graph}(\overline{\Lambda}) \in \mathcal{B}_{c}(X)\otimes \mathcal{B}(\mathbb{R}^{d})$. Now we prove $ii)$. Fix some open set $O \subset \mathbb{R}^{d}$ and let ${\Lambda}^{c}(x)= \mathbb{R}^{d} \backslash{\Lambda(x)}$. As $\mbox{Graph}(\Lambda^{c})= \left(X \times \mathbb{R}^{d}\right) \backslash{\mbox{Graph}(\Lambda)} \in \mathcal{A}(X \times \mathbb{R}^{d})$, from  \citep[Proposition 7.39 p165]{BS} we get that
$$\left\{ x \in X,\; \Lambda^{c}(x) \cap O \neq \emptyset\right\}= Proj_{X}\left((X \times O) \cap \mbox{Graph}(\Lambda^{c}) \right) \in \mathcal{A}(X) \subset \mathcal{B}_{c}(X).$$
Thus $\Lambda^{c}$ is $\mathcal{B}_{c}(X)$-measurable and
as $\Lambda^{c}$ is closed valued, \citep[Theorem 14.8 p648]{rw} applies and $\mbox{Graph}(\Lambda^{c})$ belongs to
$\mathcal{B}_{c}(X)\otimes \mathcal{B}(\mathbb{R}^{d})$ and $\mbox{Graph}(\Lambda)$ as well.
\end{proof}

\section*{Acknowledgments}
L. Carassus thanks LPMA (UMR 7599) for support.

\bibliographystyle{plainnat}

\bibliography{bibliomiklos4}

\begin{thebibliography}{33}
\providecommand{\natexlab}[1]{#1}
\providecommand{\url}[1]{\texttt{#1}}
\expandafter\ifx\csname urlstyle\endcsname\relax
  \providecommand{\doi}[1]{doi: #1}\else
  \providecommand{\doi}{doi: \begingroup \urlstyle{rm}\Url}\fi

\bibitem[Aliprantis and Border(2006)]{Hitch}
C.~D. Aliprantis and K.~C. Border.
\newblock \emph{Infinite Dimensional Analysis~: A Hitchhiker's Guide}.
\newblock Grundlehren der Mathematischen Wissenschaften [Fundamental Principles
  of Mathematical Sciences]. Springer-Verlag, Berlin, 3rd edition, 2006.

\bibitem[Avellaneda et~al.(1996)Avellaneda, Levy, and Paras]{AP95}
M.~Avellaneda, A.~Levy, and A.~Paras.
\newblock Pricing and hedging derivatives securities in markets with uncertain
  volatilities.
\newblock \emph{Applied Mathematical Finance}, 2\penalty0 (2):\penalty0 73--88,
  1996.

\bibitem[Bartl(2016)]{Bart16}
D.~Bartl.
\newblock Exponential utility maximization under model uncertainty for
  unbounded endowments.
\newblock \emph{ArXiv}, 2016.

\bibitem[Bertsekas and Shreve(2004)]{BS}
D.~P. Bertsekas and S.~Shreve.
\newblock \emph{Stochastic Optimal Control: The Discrete-Time Case}.
\newblock Athena Scientific, 2004.

\bibitem[Blanchard and Carassus(2017)]{BC17}
R.~Blanchard and L.~Carassus.
\newblock Quantitative fundamental theorem of asset pricing in discrete-time
  case with multiple priors.
\newblock \emph{in preparation}, 2017.

\bibitem[Blanchard et~al.(2016)Blanchard, Carassus, and R\'asonyi]{BCR}
R.~Blanchard, L.~Carassus, and M.~R\'asonyi.
\newblock Non-concave optimal investment and no-arbitrage: a measure
  theoretical approach.
\newblock \emph{arXiv:1602.06685}, 2016.

\bibitem[Bouchard and Nutz(2015)]{BN}
B.~Bouchard and M.~Nutz.
\newblock Arbitrage and duality in nondominated discrete-time models.
\newblock \emph{Annals of Applied Probability}, 25\penalty0 (2):\penalty0
  823--859, 2015.

\bibitem[Carassus and R\'asonyi(2016)]{CR14}
L.~Carassus and M.~R\'asonyi.
\newblock Maximization of non-concave utility functions in discrete-time
  financial market.
\newblock \emph{Mathematics of Operations Research}, 41\penalty0 (1):\penalty0
  146--173, 2016.

\bibitem[Carassus et~al.(2015)Carassus, R\'asonyi, and Rodrigues]{CRR}
L.~Carassus, M.~R\'asonyi, and A.~M. Rodrigues.
\newblock Non-concave utility maximisation on the positive real axis in
  discrete time.
\newblock \emph{Mathematics and Financial Economics}, 9\penalty0 (4):\penalty0
  325--348, 2015.

\bibitem[Castaing and Valadier(1977)]{CV77}
C.~Castaing and M.~Valadier.
\newblock \emph{Convex Analysis and Measurable Multifunctions}, volume 580.
\newblock Springer, Berlin, 1977.

\bibitem[Denis and Kervarec(2013)]{DenKer}
L.~Denis and M.~Kervarec.
\newblock Optimal investment under model uncertainty in nondominated models.
\newblock \emph{SIAM Journal on control and optimization}, 51\penalty0
  (3):\penalty0 1803--1822, 2013.

\bibitem[Denis and Martini(2006)]{DM06}
L.~Denis and C.~Martini.
\newblock A theoretical framework for the pricing of contingent claims in the
  presence of model uncertainty.
\newblock \emph{Annals of Applied Probability}, 16(2):\penalty0 827--852, 2006.

\bibitem[Denis et~al.(2011)Denis, Hu, and Peng]{DHP11}
L.~Denis, M.~Hu, and S.~Peng.
\newblock Function spaces and capacity related to a sublinear expectation:
  application to {G}-brownian motion paths.
\newblock \emph{Potential Analysis}, 34 (2)\penalty0 (139-161), 2011.

\bibitem[Donoghue(1969)]{WFD}
W.~F. Donoghue.
\newblock Distributions and fourier transforms.
\newblock \emph{Vol. 32 of Pure and Applied Mathematics. Elsevier.}, 1969.

\bibitem[F{\"o}llmer and Schied(2002)]{fs}
H.~F{\"o}llmer and A.~Schied.
\newblock \emph{Stochastic Finance: An Introduction in Discrete Time}.
\newblock Walter de Gruyter \& Co., Berlin, 2002.

\bibitem[F{\"o}llmer et~al.(2009)F{\"o}llmer, Schied, and Weber]{FSW09}
H.~F{\"o}llmer, A.~Schied, and S.~Weber.
\newblock Robust preferences and robust portfolio choice.
\newblock \emph{Mathematical Modelling and Numerical Methods in Finance}, 2009.

\bibitem[Gelbaum and Olmsted(1964)]{GO1}
B.R. Gelbaum and J.H. Olmsted.
\newblock \emph{Counterexamples in Analysis}.
\newblock Dover Publications, Inc., 1964.

\bibitem[Gilboa and Schmeidler(1989)]{Gilb}
I.~Gilboa and D.~Schmeidler.
\newblock Maxmin expected utility with non-unique prior.
\newblock \emph{Journal of Mathematical Economics}, 18(2):\penalty0 141--153,
  1989.

\bibitem[Jacod and Shiryaev(1998)]{JS98}
J.~Jacod and A.~N. Shiryaev.
\newblock Local martingales and the fundamental asset pricing theorems in the
  discrete-time case.
\newblock \emph{Finance Stochastic}, 2:\penalty0 259--273, 1998.

\bibitem[Knight(1921)]{Kni}
F.~Knight.
\newblock \emph{Risk, Uncertainty, and Profit}.
\newblock Boston, MA: Hart, Schaffner Marx; Houghton Mifflin Co, 1921.

\bibitem[Kramkov and Schachermayer(1999)]{KS99}
D.~O. Kramkov and W.~Schachermayer.
\newblock The asymptotic elasticity of utility functions and optimal investment
  in incomplete markets.
\newblock \emph{Annals of Applied Probability}, 9:\penalty0 904--950, 1999.

\bibitem[Lyons(1995)]{Ly95}
F.~Lyons.
\newblock Uncertain volatility and the risk-free synthesis of derivatives.
\newblock \emph{Journal of Applied Finance}, 2:\penalty0 117--133, 1995.

\bibitem[Neufeld and Sikic(2016)]{NS16}
A.~Neufeld and M.~Sikic.
\newblock Robust utility maximization in discrete time with friction.
\newblock \emph{ArXiv}, 2016.

\bibitem[Nutz(2016)]{Nutz}
M.~Nutz.
\newblock Utility maximisation under model uncertainty in discrete time.
\newblock \emph{Mathematical Finance}, 26\penalty0 (2):\penalty0 252--268,
  2016.

\bibitem[Pennanen and Perkkio(2012)]{PennanenPerkio}
T.~Pennanen and A-P. Perkkio.
\newblock Stochastic programs without duality gaps.
\newblock \emph{Mathematical Programming}, 136:\penalty0 91--220, 2012.

\bibitem[R\'asonyi and Stettner(2005)]{RS05}
M.~R\'asonyi and L.~Stettner.
\newblock On the utility maximization problem in discrete-time financial market
  models.
\newblock \emph{Annals of Applied Probability}, 15:\penalty0 1367--1395, 2005.

\bibitem[R\'asonyi and Stettner(2006)]{RS06}
M.~R\'asonyi and L.~Stettner.
\newblock On the existence of optimal portfolios for the utility maximization
  problem in discrete time financial models.
\newblock \emph{In: Kabanov, Y.; Lipster, R.; Stoyanov,J. (Eds), From
  Stochastic Calculus to Mathematical Finance, Springer.}, pages 589--608,
  2006.

\bibitem[Rockafellar(1970)]{cvx}
R.~T. Rockafellar.
\newblock \emph{Convex Analysis}.
\newblock Princeton, 1970.

\bibitem[Rockafellar and Wets(1998)]{rw}
R.~T. Rockafellar and R.~J.-B. Wets.
\newblock \emph{Variational analysis}.
\newblock Grundlehren der Mathematischen Wissenschaften [Fundamental Principles
  of Mathematical Sciences]. Springer-Verlag, Berlin, 1998.
\newblock ISBN 3-540-62772-3.

\bibitem[Sainte-Beuve(1974)]{bv}
M.-F. Sainte-Beuve.
\newblock On the extension of von {N}eumann-{A}umann's theorem.
\newblock \emph{Journal of Functional Analysis}, 17\penalty0 (1):\penalty0
  112--129, 1974.

\bibitem[Schachermayer(2001)]{S01}
W.~Schachermayer.
\newblock Optimal investment in incomplete markets when wealth may become
  negative.
\newblock \emph{Annals of Applied Probability}, 11:\penalty0 694--734, 2001.

\bibitem[Sierpinski(1920)]{Siep}
W.~Sierpinski.
\newblock Sur un probl{\`e}me concernant les ensembles mesurables
  superficiellement.
\newblock \emph{Fundamenta Mathematica}, 1:\penalty0 p112--115, 1920.

\bibitem[von Neumann and Morgenstern(1947)]{vNM}
J.~von Neumann and O.~Morgenstern.
\newblock \emph{Theory of games and economic behavior}.
\newblock Princeton University Press, 1947.

\end{thebibliography}

\end{document}